\documentclass[11 pt, onecolumn, draftcls]{IEEEtran}
\usepackage{amsbsy,color,multicol}
\usepackage{floatflt} % text flows around figures

\usepackage{amsmath}
\usepackage{amssymb}
\usepackage{times}
\usepackage{graphicx}
\usepackage{xspace}
\usepackage{paralist} % compact and in-paragraph* lists
\usepackage{setspace} % buggy somehow
\usepackage{xypic}
\xyoption{curve}
\usepackage{latexsym}
\usepackage{theorem}
\usepackage{ifthen}
\usepackage{caption}
\usepackage{subcaption}
\usepackage{makeidx,multirow,textcomp,longtable,afterpage,setspace}
\usepackage[english]{babel}
\usepackage[latin1]{inputenc}
\usepackage[ruled,vlined]{algorithm2e}
\usepackage{booktabs}
\usepackage{mathtools}
\newlength\myindent
\newlength\mycolwid

{\theoremheaderfont{\it} \theorembodyfont{\rmfamily}
\newtheorem{theorem}{Theorem}

\newtheorem{definition}[theorem]{Definition}
\newtheorem{example}[theorem]{Example}

}

\title{Epidemics in Multipartite Networks: Emergent Dynamics}

\author{Augusto Santos$^{*}$, Jos\'e M.~F.~Moura$^{\natural}$, and Jo\~{a}o M. F. Xavier$^\dagger$
\thanks{$*$ A.~Santos is with the Dep.~of Electrical and Computer Engineering, Carnegie Mellon University, USA, and Instituto de Sistemas e Robotica (ISR), Instituto Superior T\'{e}cnico (IST), Av.~Rovisco Pais, Lisboa, Portugal (augustos@andrew.cmu.edu).}
\thanks{$^\natural$J.~M.~F.~Moura is with the Dep.~of Electrical and Computer Engineering,
Carnegie Mellon University, Pittsburgh, PA, USA 15213, ph:(412)268-6341, fax: (412)268-3890 (moura@ece.cmu.edu).}
\thanks{$\dagger$ J.~M.~F.~Xavier is with ISR, IST, Av.~Rovisco Pais, Lisboa, Portugal (jxavier@isr.ist.utl.pt).}
\thanks{The work of Jos\'e M.~F.~Moura and Augusto Santos was supported in part by the National Science Foundation under Grant \#~CCF--1011903 and in part by the Air Force Office of Scientific Research under Grant \#~FA--95501010291.

The work of J.~M.~F.~Xavier and A.~Santos was also supported by the Funda\c{c}\~{a}o para a Ci\^{e}ncia e a Tecnologia (FCT) (Portuguese Foundation for Science and Technology) through the Carnegie Mellon$|$Portugal Program under Grant SFRH/BD/33516/2008, CMU-PT/SIA/0026/2009 and SFRH/BD/33518/2008, and by ISR/IST pluriannual funding (POSC program, FEDER).}}

\begin{document}
\maketitle \thispagestyle{empty} \maketitle

%%%%%%%%%%%%%%%%%%%%%%%%%%%%%%%%%%%%%%%%%%%%%%%%%%%%%%%%%%%%%%%%%%%%%%%%%%%%%%%%%%%%%%%%%%%%%%%%%%%%%%%%%%%%%%%%%%%%

%\pagestyle{plain}
%\newpage%\date{}

\begin{abstract}

Single virus epidemics over \textit{complete} networks are widely explored in the literature as the fraction of infected nodes is, under appropriate microscopic modeling of the virus infection, a Markov process. With non-complete networks, this macroscopic variable is no longer Markov. In this paper, we study virus diffusion, in particular, multi-virus epidemics, over \textit{non}-complete stochastic networks. We focus on multipartite networks. In companying work~\cite{augusto_moura_emergent}, we show that the peer-to-peer local random rules of virus infection lead, in the limit of large multipartite networks, to the emergence of structured dynamics at the \emph{macroscale}. The \textbf{exact} fluid limit evolution of the fraction of nodes infected by each virus strain across \emph{islands} obeys a set of nonlinear coupled differential equations, see~\cite{augusto_moura_emergent}. In this paper, we develop methods to analyze the qualitative behavior of these limiting dynamics, establishing conditions on the virus micro characteristics and network structure under which a virus persists or a natural selection phenomenon is observed.
\end{abstract}


%\newpage
%\pagenumbering{roman}
%\setcounter{page}{5}
%------------------------------------------------------------
%------------------------------------------------------------
%------------------------------------------------------------

% Acknowledgements
%\include{acknowledgements/acknowledgements}

% Abstract
%\include{abstract}

%------------------------------------------
% Contents / Tables / Figures
%------------------------------------------
%\tableofcontents \pagebreak
%\listoftables \pagebreak
%\listoffigures \pagebreak

%------------------------------------------
% Acronyms
%------------------------------------------
%\chapter*{Acronyms}
%\linespread{1.5}{\begin{tabular}{lll}
%\textbf{ANSIG}  					& \textbf{-}  & \emph{Analytic signature}. \\
%\end{tabular}}

%\setcounter{page}{2} \pagenumbering{arabic}

%\input{Chapter1_auto_M}
%\setcounter{page}{2} \pagenumbering{arabic}

\section{Introduction}\label{sec:introducao}
This paper studies the \textit{macroscopic scale} dynamics of a multi-virus epidemics or diffusion over large stochastic \textit{non-complete} networks of agents. Questions of interest include when a virus persists, when among multiple strains of virus we observe survival of the fittest, or what is the distribution of the fraction of infected agents over the various strains of virus in the network. These are well studied when the network is \textit{complete}, i.e., any agent interacts directly with any other agent, and a vast body of literature describes the dynamics of the fraction of infected nodes by nonlinear ordinary differential equations~(ODEs) that are arrived at through conservation or full mixing arguments, \cite{Daley}. These nonlinear ODEs can also be rigorously derived because the fraction of infected nodes in the complete network is a Markov process under the standard independence assumptions on the peer-to-peer (\textit{microscopic}) infection process, and the resulting \textit{macroscopic} or \textit{global behavior} of the epidemics is the fluid limit of this Markov process as the size of the complete network grows to infinity, see \cite{Antunes,Antunes2}. When the network is not complete, the fraction of network infected nodes is no longer Markov and studying the network global or macroscopic behavior is a major challenge. Attempts to overcome these difficulties make unsupported or unrealistic assumptions like the independence of the (random) states of infection of neighboring agents, \cite{Mieghem}. In~\cite{augusto_moura_emergent}, we derive, from a basic microscopic SIS -- susceptible-infected-susceptible -- infection model and without making unrealistic simplifying assumptions, the mean field ODEs describing the global behavior of epidemics for a class of non-complete stochastic networks, namely, multipartite networks. The resulting mean field equations are nonlinear coupled ODEs. This paper studies the qualitative behavior of these mean field ODEs, i.e., the stability of their equilibria dynamics, to establish the emergent network macroscopic behaviors. Their coupled nonlinear behavior defies the use of Lyapunov methods. We develop a new methodology that upper- and lower-bounds the limiting dynamics of the stochastic network by the much simpler to analyze dynamics of first order nonlinear systems. We consider single- and multi-virus epidemics and arbitrary regular multipartite networks. This paper, together with~\cite{augusto_moura_emergent}, derives rigorously from basic peer-to-peer principles of diffusion the characterization of the global diffusion or infection behavior in multipartite networked systems in the limit of large systems. We believe this to be the first \emph{microscopic-to-macroscopic} study that goes beyond complete networked systems to obtain the \textbf{exact} impact of a non-complete topology on global infection and diffusion dynamics.

%Subsection~\ref{subsec:microscopicinfection}Subsection~\ref{subsec:dynamicsmultipartite} Subsection~\ref{subsec:meanfield}
%Subsection~\ref{subsec:bipartitenet-singlevirus} Subsection~\ref{subsec:bipartitenet-bivirus} Subsection~\ref{subsec:multipartitenet-singlevirus} Subsection~\ref{subsec:multipartitenet-bivirus}

\textbf{Summary of the paper.} Section~\ref{sec:problemsetup} sets-up the model of microscopic epidemics, describes the multipartite network topology, and recalls the mean field equations in~\cite{augusto_moura_emergent} governing the limiting dynamics. Section~\ref{sec:qualitativeanalysis} establishes the qualitative behavior of the limiting dynamics for single and bi-viral epidemics in a bipartite network. Section~\ref{sec:multipartite} extends these results to arbitrary general regular multipartite networks. Concluding remarks are in Section~\ref{sec:asymmetry}.

%%%%%%%%%%%%%%%%%%%%%%%%%%%%%%%%%%%%%%%%%%%%%%%%%%%%%%%%%%%%%%%%%%%%%%%%%%%%%%%%%%%%%%%%%%%%%%%%%%%%%%%%%%%%%%%%%%%%%%%%%%%%%%%%%%%%%%%%%%%%%%%%%%%%%%%%%%%%%

\section{Problem Setup}
\label{sec:problemsetup}
%\hspace{0.43cm}
%

This section presents the underlying stochastic network model for the peer-to-peer virus infection and the mean field equations describing the macroscopic epidemics dynamics in the limit of large networks established in~\cite{augusto_moura_emergent}.

 The environment where \emph{actions} take place is a network modeled as an undirected simple graph (no self-loops) $G=\left(V,E\right)$, where $V\subset \mathbb{N}$ is the set of nodes and $E=\left\{\left\{i,j\right\}\,:\,i,\,j\in V\mbox{ and }i\neq j\right\}$ is the set of edges. We write $i\sim j$ if~$i$ and~$j$ are neighbors, i.e., $\left\{i,j\right\}\in E$. The number of nodes is $|V|=N$.
On this network, the infection or diffusion process $\left(\mathbf{X}^N(t)\right)$ is the microstate of the network that collects the \emph{state} of each node $i\in V$ for every $t$, $t\geq 0$.
% For a single virus, the microstate is an $N$-dimensional binary valued vector, $\mathbf{X}^N(t)\in \left\{0,1\right\}^N$. More generically, for a $k$-virus epidemics, the microstate is a $N\times k$ binary valued matrix.
%
%
%The terms \emph{state} and \emph{actions} are made precise next.
%
\subsection{Microscopic infection model: Susceptible-infected-susceptible (SIS)}
\label{subsec:microscopicinfection}
We assume that all stochastic processes are supported in a single probability space $\left(\Omega,\mathcal{F},\mathbb{P}\right)$.

\textit{State}: With single virus, the microstate of the network is an~$N$-dimensional vector state~$\mathbf{X}^N(t)$ where its $i$th-component $X^N_i(t,\omega)$ at time $t\geq 0$ and for the realization $\omega\in\Omega$ can be in one of two states, i.e., it is binary valued: $X^N_i(t,\omega)=1$ if node~$i$ is infected (or contaminated), and $X^N_i(t,\omega)=0$ if it is healthy. These are the only two possible states. For $K$~multiple strains of virus, the microstate is a $N\times K$ matrix $\mathbf{X}^N(t)\in\left\{0,1\right\}^{N\times K}$ where the rows index the nodes and the columns index the strains. A node $i\in V$ is infected with strain $k\in \left\{1,\ldots,K\right\}$ at time $t$, $t\geq 0$, if $X^N_{ik}(t)=1$, and we say that node $i$ is $k$-infected. Node $i\in V$ is healthy at time $t$, $t\geq 0$, if $X^N_{ik}(t)=0$ for all $k\in \left\{1,\ldots,K\right\}$. The microstate is the mapping $\mathbf{X}^N$ summarized as:
\begin{align*}
%\mathbf{X}^N:\:\mathbb{R}_{+}\times \Omega & \longrightarrow \left\{0,1\right\}^{N\times K}\\
%\left(t,\omega\right)   & \longmapsto  \mathbf{X}^N(t,\omega).
\mathbf{X}^N:\:\mathbb{R}_{+}\times \Omega & \longrightarrow \left\{0,1\right\}^{N\times K}, \:\:
\left(t,\omega\right)   \longmapsto  \mathbf{X}^N(t,\omega).
\end{align*}
%The vector process $\left(\mathbf{X}^N(t)\right)$ is referred to as the microstate of the network.

\textit{Local exclusion principle}: At any $t\geq 0$, a node may only be infected by a single strain. If a node~$i$ infected by strain~$k_1$ heals at $t=t^*\geq 0$, it may be infected by strain~$k_2$ at $t^\dag>t^*$. The rows in the microstate $\mathbf{X}^N(t)$ are either zero rows or have a single nonzero entry, which is a~$1$.

\textit{Actions}: We assume a \textit{susceptible-infected-susceptible} (SIS) model. When a node is infected, it is a matter of time to either contaminate its one-hop peers or to heal. We describe both the time for infection and for healing as independent exponentially distributed random variables. More specifically, each node has $1+K$ independent clocks, one for healing and the other~$K$ clocks for the corresponding $k$-infection. Once a node is $k$-infected, all clocks, for healing and for $k$-infection, are activated and will ring after exponentially distributed random times. If the healing clock of an infected node rings first, say the healing clock of infected node~$i$ rings, node~$i$ heals. If the clock for the $k$-infection of any of the infected nodes rings first, say for infected node~$i$, node~$i$ infects a uniformly randomly chosen neighbor with virus strain~$k$. If the chosen peer is already infected (with any strain), by the local exclusion principle, the network microstate $\left(\mathbf{X}^N(t)\right)$ stays unchanged.  Thus, our building block is a sequence of independent, identical, exponentially distributed random variables $T^{c}_n\sim {\sf Exp}\left(\gamma_k\right)$ (superindex~$c$ for contamination) and $T^{h}_n\sim {\sf Exp}\left(\mu_k\right)$ (superindex~$h$ for healing). The parameters $\gamma_k$ and $\mu_k$ are the rates of $k$-infection and healing, indexed by the underlying strains. A strain of virus is characterized by a pair $\left(\gamma_k, \mu_k\right)$ and two strains~$1$ and~$2$ are different if $\left(\gamma_1, \mu_1\right)\neq\left(\gamma_2, \mu_2\right)$.
\begin{figure} [hbt]
\begin{center}
\includegraphics[scale= 0.3]{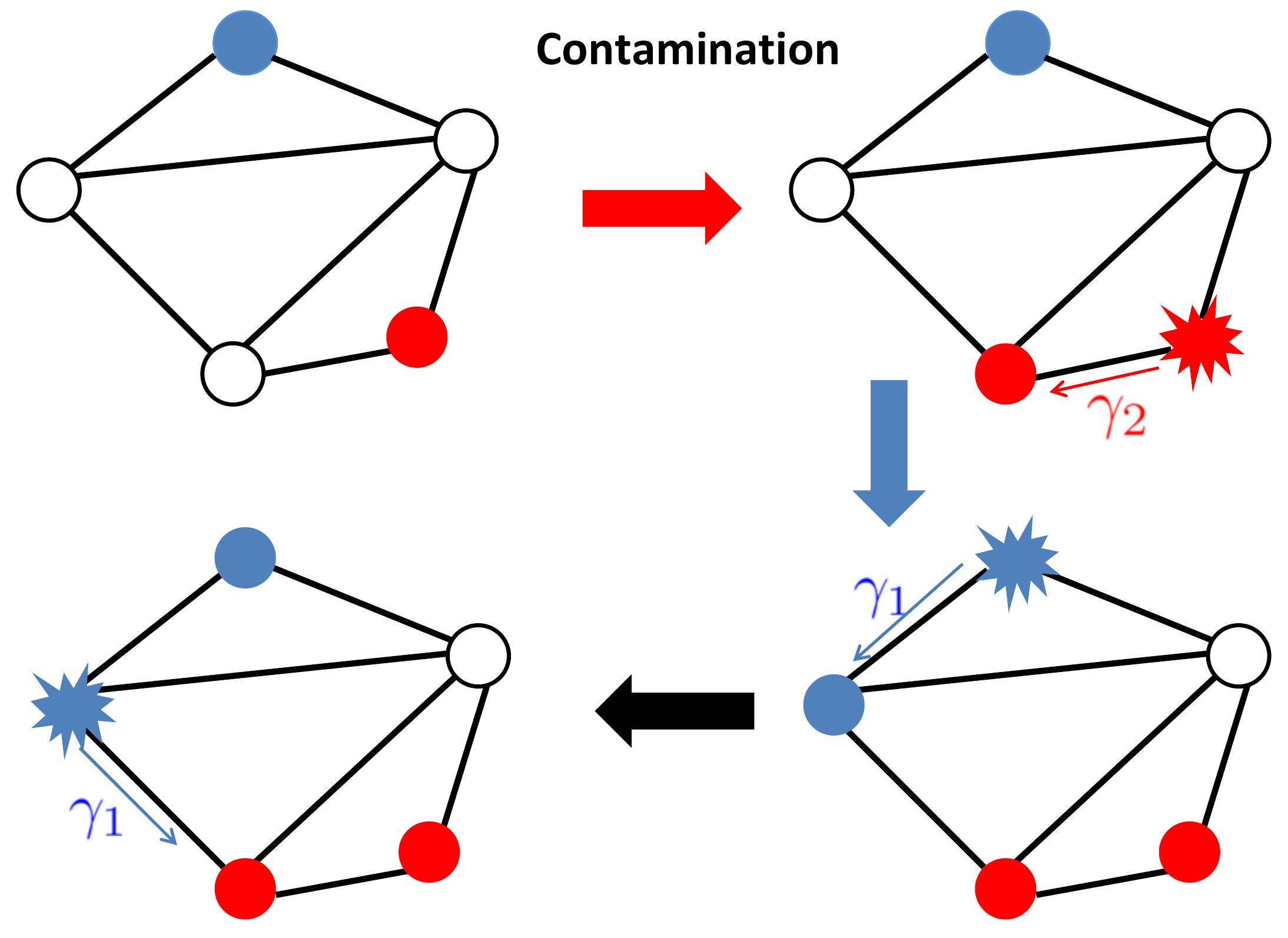}
\caption{Viral strains are characterized by transmission and healing rates. Once the clock for infection rings, a neighbor is randomly picked to be infected, unless it is already infected. The bottom left picture emphasizes the local exclusion principle.}\label{fig:Contamination22}
\end{center}
\end{figure}
Figure~\ref{fig:Contamination22} illustrates the dynamics for a two virus infection and the local exclusion principle--an infected node is not infected by another virus before healing first.

From the microscopic description of the law of evolution, the microstate $\left(\mathbf{X}^N(t)\right)$ is a Markov process.
 %whose dynamics are characterized by its transition rate matrix $\mathbf{Q}^N$.
 For a single virus with parameters $(\gamma,\mu)$, the generic entry $Q\left(\mathbf{X}^N(t),\mathbf{X}^N(t)+\mathbf{v}\right)$ of its transition rate matrix $\mathbf{Q}^N\in\mathbb{R}^{2^N\times2^N}$ is:
  %its generic entry $Q\left(\mathbf{X}^N(t),\mathbf{X}^N(t)+\mathbf{v}\right)$ is:
{\allowdisplaybreaks
\begin{align*}
%$$
%\align
Q\left(\mathbf{X}^N(t),\mathbf{X}^N(t)+\mathbf{e}_i\right) & = \gamma\sum_{j\sim i}X^N_j(t)\frac{1}{d_j},\phantom{0} \mbox{$i$ healthy}
\\
Q\left(\mathbf{X}^N(t),\mathbf{X}^N(t)-\mathbf{e}_i\right) & = 0, \phantom{\gamma\sum_{j\sim i}X^N_j(t)\frac{1}{d_j}} \mbox{$i$ healthy}\\
Q\left(\mathbf{X}^N(t),\mathbf{X}^N(t)+\mathbf{e}_i\right) & = 0, \phantom{\gamma\sum_{j\sim i}X^N_j(t)\frac{1}{d_j}} \mbox{$i$ infected}\\
Q\left(\mathbf{X}^N(t),\mathbf{X}^N(t)-\mathbf{e}_i\right) & = \mu, \phantom{\gamma\sum_{j\sim i}X^N_j(t)\frac{1}{d_j}} \mbox{$i$ infected}\\
Q\left(\mathbf{X}^N(t),\mathbf{X}^N(t)-\mathbf{v}\right) & = 0, \phantom{\gamma\sum_{j\sim i}X^N_j(t)\frac{1}{d_j}} \mbox{$||\mathbf{v}||_1>1$,}
\end{align*}
}
%\endalign
%$$
\hspace{-2mm}
where $d_j$ is the degree or number of neighbors of node~$j$ and $\mathbf{e}_i\in\mathbb{R}^{N}$ is the canonical vector with all entries equal to zero except the $i$th entry that is~$1$. For $K$-virus, the rate $\mathbf{Q}^N$ is a tensor; its generic element is a straightforward generalization of the generic entry of the rate matrix~$\mathbf{Q}^N$ for a single virus. In the sequel, we usually consider explicitly the single virus epidemics, but still refer to~$\mathbf{Q}^N$ as the rate or rate matrix, even if we study a $K$-virus epidemics.

\textit{Network macrostate}: The rate matrix $\mathbf{Q}^N$ is too large even for moderate size networks. To address this curse of dimensionality, we rely on low-dimensional network state statistics $\mathbf{Y}=f(\mathbf{X})$, where $f:\mathbb{R}^{N\times K}\rightarrow \mathbb{R}^M$ is a measurable function and $M<<N$. The stochastic process $\mathbf{Y}=\left(\mathbf{Y}^N(t)\right)$ is referred to as a \textit{macrostate} of the network. One macrostate of particular interest throughout this paper is the \textit{fraction of infected nodes} of strain~$k$:
\begin{align*}
Y_{k}^{N}(t)=\sum_{i=1}^{N} X^{N}_{ik}(t), &\,\,\,
\overline{Y}_{k}^{N}(t)=\frac{Y_{k}^{N}(t)}{N},  \nonumber
\end{align*}
where $Y_k^N(t)$ and $\overline{Y}_k^N(t)$ represent the number and fraction, respectively, of nodes infected by strain~$k$ in the $N$-network at time $t$, $t\geq 0$. For single virus, we write $Y^N=\left(Y^N(t)\right)$ and $\overline{Y}^N=\left(\overline{Y}^N(t)\right)$, dropping the superindex~$N$ when clear from the context.

%We can read off each realization $\omega\in\Omega$ as the history of all clocks triggered for all time $t\geq 0$. For a network with $N$ nodes, we say that node $i$ is $k$-infected -- infected with the viral strain $k$ -- at time $t$ if it is contaminated with virus type $k$ and we represent it as $X^{\left(N\right)}_{ik}(t)=1$ (or $X^{\left(N\right)}_{ik}(t)=0$, if not infected.) as the number and fraction, respectively, of $k$-infected nodes in the $N$-network.

%\subsection{Complete vs General Networks}

\begin{example}[Complete network]\label{exp:completenetwork}
For a \textit{complete} network, each node can infect any other node. For single virus, the transition rate of $\left(Y^N(t)\right)$ depends solely on itself, e.g., \cite{paper:CDC}, and this macrostate is Markov. To study its dynamics, we need its one-dimensional transition rate instead of the transition rates for the full $2^N$ microstate. We have:
\begin{eqnarray}
Q\left(Y^N(t),Y^N(t)+1\right) & = & \gamma Y^N(t)\left(N-Y^N(t)\right)\nonumber\\
Q\left(Y^N(t),Y^N(t)-1\right) & = & \mu Y^N(t).\nonumber
\end{eqnarray}
 Complete networks are well studied in the micro-to-macro network diffusion. For instance, Reference~\cite{Antunes2} considers a multiclass flow of packets on a complete network. Starting from its microscopic statistics, it shows that the empirical distribution of nodes across the possible configurations of packets at each node $\left(\mathbf{\overline{Y}}^N(t)\right)$ is Markov, then it proves that the process converges weakly to the solution of an ordinary differential equation (ODE) as the number of nodes grows large, and provides the qualitative analysis of the resulting ODE.
\hfill$\small \blacksquare$
\end{example}

To handle arbitrary topologies is much more challenging because, for a general network topology, the macrostate process $\left(\mathbf{Y}^N(t)\right)$ is not Markov as we show with the following example.

\begin{example}[Arbitrary network] \label{exp:arbitrarynetwork} Consider the two microstate configurations in Figure~\ref{fig:contaminationclu} for a single virus cycle network $C_6$, where the darkened (colored) nodes represent infected nodes.
\begin{figure} [hbt]
\begin{center}
\includegraphics[scale= 0.3]{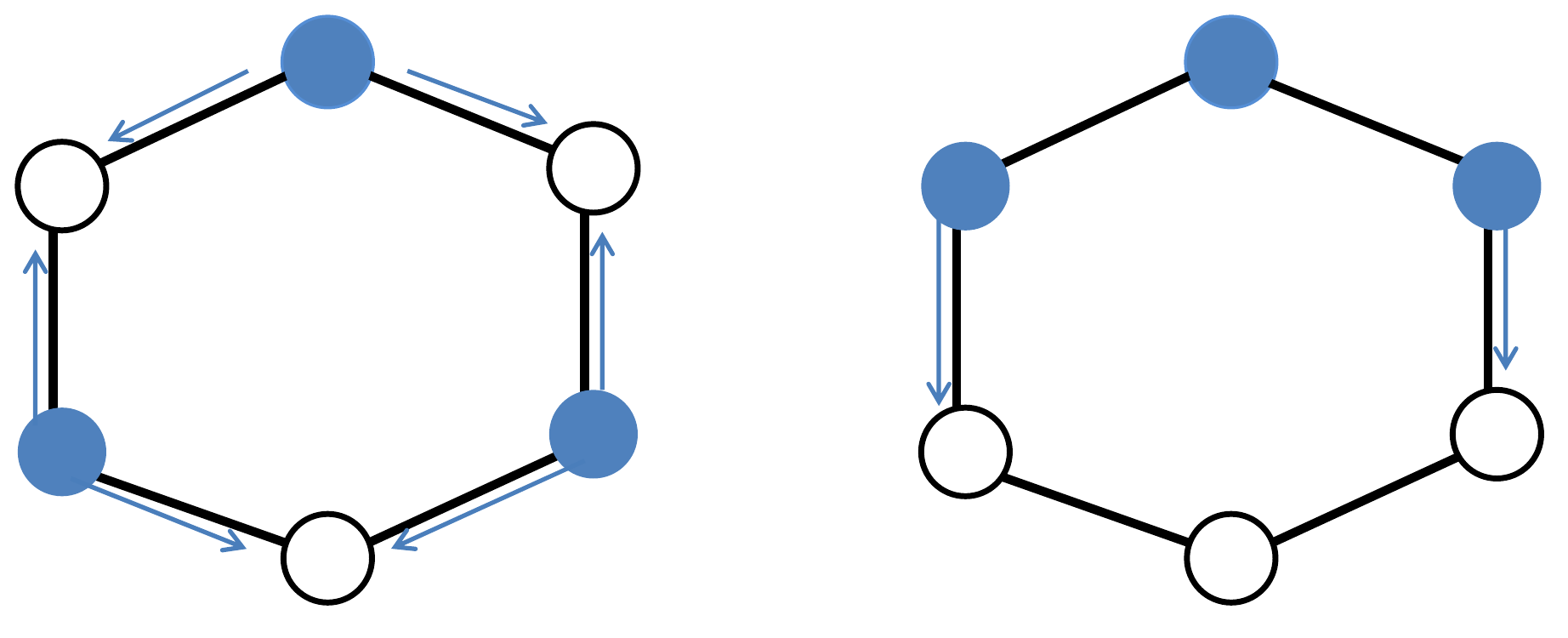}
\caption{Two distinct microstates for a single virus cycle network $C_6$ with the same number of infected nodes, but leading to different cross-transition rates.}\label{fig:contaminationclu}
\end{center}
\end{figure}
We show that the rates to increase the number of infected nodes process $\left(Y^N(t)\right)$ are coupled with the microstate $\left(\mathbf{X}^N(t)\right)$. Indeed, the clustered configuration on the right yields a lower rate as the potential infections can come only from its boundary nodes, whereas in the configuration on the left, any neighbor of an infected node can be contaminated. In words, the rates at time $t$ are not uniquely determined by $Y^N(t)$ and depend on the microstate, or, more formally, they are not adapted to the natural filtration, the $\sigma$-algebra $\sigma\left(Y^N(s),s\leq t\right)$.
\hfill$\small \blacksquare$
\end{example}

Example~\ref{exp:arbitrarynetwork} shows that the dynamics in arbitrary networks are much more challenging than in complete networks. Reference~\cite{Mieghem} considers non complete topologies, bypassing the coupling problem illustrated in Example~\ref{exp:arbitrarynetwork} by replacing the exact rates of transition of the microstate $\left(\mathbf{X}^N(t)\right)$ by their average. If the states of the nodes, i.e., the scalar entries of the microstate, were independent processes, the approximation would be accurate for large networks. But this is not the case as the authors themselves point out. Similar approaches replacing rates by their averages are standard with non complete networks. Another example, representative of many epidemics and diffusion macroscopic studies, is~\cite{Pastor-Satorras-Vespignani-2001} that adopts it by neglecting the correlation among infected nodes when studying SIS epidemics in scale free networks.

In summary, the curse of dimensionality has been studied under one of the following settings:
\begin{enumerate}
\item \textit{Bottom up over a complete network}. From the microstatistics of the diffusion, the low-dimensional process $\left(Y^N(t)\right)$ is shown to be Markov--see \cite{Antunes2,Antunes}.
\item \textit{Bottom up with relaxation}. Since the network topology is arbitrary, $\left(Y^N(t)\right)$ is no longer Markov. The non-Markovianity is bypassed by relaxing the micro model--neglecting correlations among microstates, or estimating bounds on the rates, see for example~\cite{Mieghem,particle}.
\item \textit{Prescribed mean field models}. The dynamics of virus diffusion as a function of global topological features are designed at the macroscale assuming average rates and neglecting correlations among nodes, see \cite{Pastor-Satorras-Vespignani-2001,Daley,Jackson} for several  models common in the literature.
\end{enumerate}

To go beyond complete networks, we introduce multipartite networks in the next Subsection~\ref{subsec:dynamicsmultipartite} and, in Subsection~\ref{subsec:meanfield}, the mean field virus dynamics derived in~\cite{augusto_moura_emergent}.

\subsection{Multipartite Networks}\label{subsec:dynamicsmultipartite}
 Multipartite networks may model networks of cities or local area networks connected by gateways.
\begin{definition}[Multipartite network]\label{def:multipartitenetwork}
A network $G=\left(V,E\right)$ is multipartite if there exists a partition $\overline{V}=\left\{V_1,\ldots,V_M\right\}$ of $V$ such that $\left\{a,b\right\}\notin E$ for any $a,\,b\in V_i$ for any $i\in\left\{1,\ldots,M\right\}$. Moreover, the following condition holds true. With $i\neq j$:
\begin{align*}
u\in V_i,\,v\in V_j,\,u\sim v\Rightarrow w\sim r,\,\,\,\, \forall\, {w\in V_i,\,r\in V_j}.
\end{align*}
 When $M=2$, the multipartite network has only two islands and is called \textit{bipartite}.
\hfill$\small \blacksquare$
\end{definition}
In the sequel, $V$ is partitioned as $\overline{V}=\left\{V_1,\ldots,V_M\right\}$. The elements $V_i$ of the partition are \textit{islands} or \textit{supernodes}. The \textit{size} of each island~$V_i$ is its cardinality $N_i=\left|V_i\right|$. The vector $\mathbf{N}=\left(N_1,\ldots,N_M\right)$ stacks the sizes $N_i$ of all~$M$ islands. If the islands are evenly sized, $N=\left|V_i\right|$, $1\leq i\leq M$, the multipartite network is \textit{symmetric}.
 By definition~\ref{def:multipartitenetwork}, if two nodes of different islands $U$ and $V$ are connected, then any node of~$U$ is connected with any node of~$V$. In this case, we say~$U$ and~$V$ are \emph{connected}, writing $U\sim V$. This abstracts the supernetwork or supergraph topological structure of the islands. Figure~\ref{fig:multipartite} depicts a symmetric multipartite network.
\begin{figure} [hbt]
\begin{center}
\includegraphics[scale= 0.3]{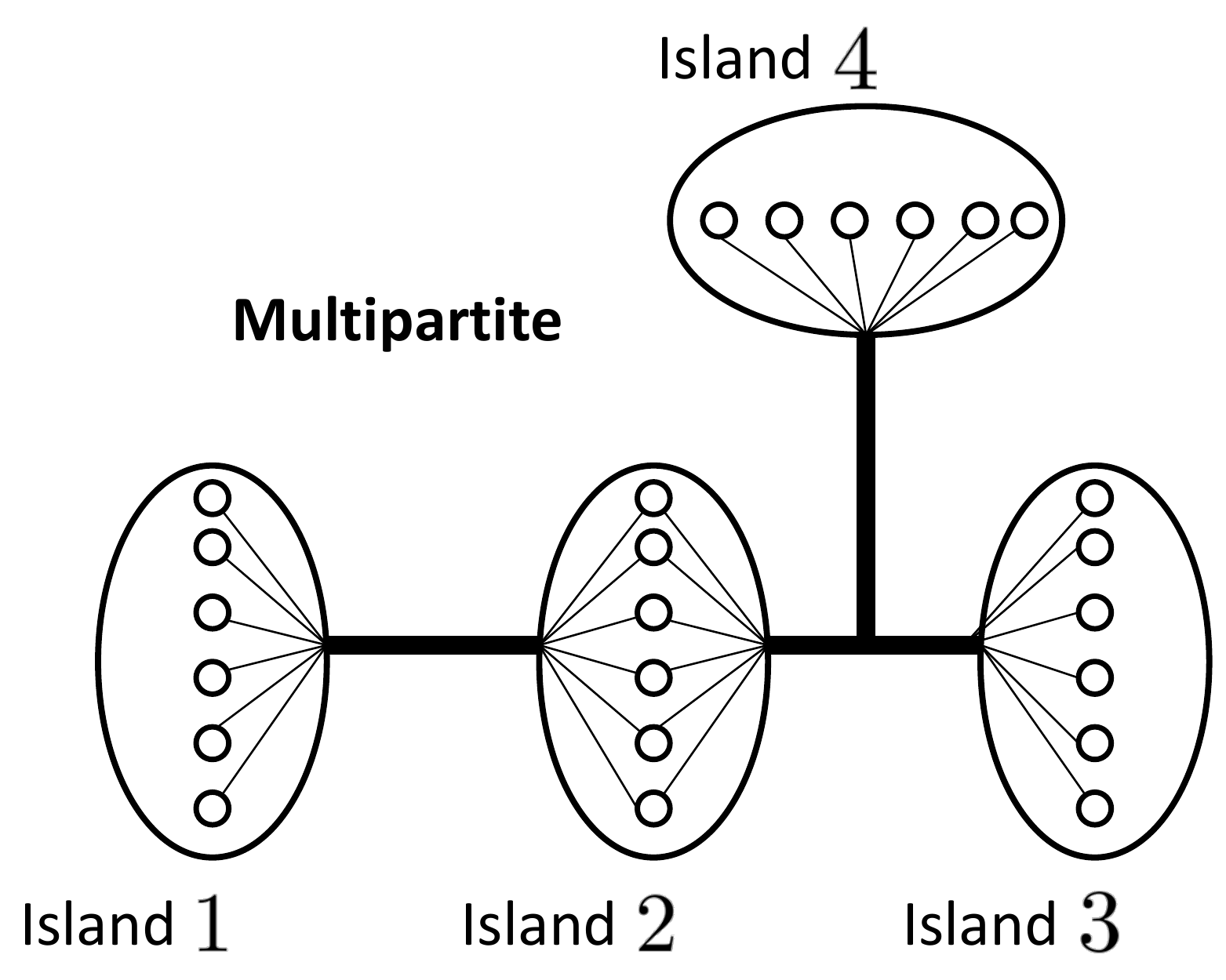}
\caption{Multipartite network: supernetwork of inter-islands; no intra-island communication.}\label{fig:multipartite}
\end{center}
\end{figure}
\begin{definition}[Superneighborhood]  The superneighborhood of $U$ is
\label{def:superneighborhood}
\begin{align*}
\mathcal{N}(U)=\left\{V\in \overline{V}\,:\,V\sim U\right\}.
\end{align*}
The degree of island~$U$ in the supernetwork, or \textit{superdegree} of~$U$, is $d_U=\left|\mathcal{N}\left(U\right)\right|$.
\hfill$\small \blacksquare$
\end{definition}
A multipartite network is \textit{regular} if all islands have the same  superdegree.

We adapt the SIS microscopic model of diffusion described in the previous Subsection~\ref{subsec:microscopicinfection} to multipartite networks. We define the binary \textit{tensor} or \textit{hypermatrix} microprocess $\left(\mathbf{X}^N(t)\right)$ as collecting the state of each node over time in the multipartite network. The entry $X^N_{ijk}(t)=1$ if node~$i$ at island~$j$ is infected at time~$t$, $t\geq 0$, with virus strain $k\in\left\{1,2,\cdots,K\right\}$, and $X^N_{ijk}(t)=0$ if the node~$i$ of island~$j$ is healthy or infected with a different strain. If only one strain of virus is present in the network, then, for notational simplicity, we suppress the extra index~$k$ and write simply $X^N_{ij}(t)$. Our SIS microscopic infection model of diffusion is set at the node level and goes as follows. Once a node~$i$ in island~$U$ is $y$-infected, it transmits the infection to a randomly chosen node in a randomly chosen neighbor island~$V$ after an exponentially distributed random time $T_{UV}^c\sim{\sf Exp}\left(\gamma^y_{UV}\right)$, if at that time the node is still infected. If the chosen node at island~$V$ is already infected, then nothing happens. Also, an $y$-infected node heals after a random time $T^h\sim{\sf Exp}\left(\mu^y\right)$. All time service random variables are assumed to be independent and have support in a single probability space $\left(\Omega,\mathcal{F},\mathbb{P}\right)$.

%With a single strain virus, we attach $d_U$ independent exponentially distributed (infection) clocks to each node $u\in U$:
%\begin{equation}
%T_{UV}^c\sim{\sf Exp}\left(\gamma\right),\mbox{ for all \:}V\in\mathcal{N}(U).
%\end{equation}
%  These clocks trigger potential infections from nodes in~$U$ towards nodes in~$V$. That is, a node $u\in U$ has as many clocks as the number of superneighbors of $U$. Once it is infected and one of its clocks $T_{UV}^c(u)$ rings, a node $v\in V$ is randomly chosen to be infected. If this node~$v$ is not infected, it becomes infected; if it is infected, nothing happens. The healing is as before.
%
%  The number of infected nodes $\left(Y^{\mathbf{N}}(t)\right)=\left(Y^{N_1}_1(t)+\ldots+Y^{N_M}_M(t)\right)$ is a Markov process for a complete network, but it fails to be Markov in the multipartite network.
%
\begin{example}[Bipartite network]\label{exp:bipartitenetwork}
We compute the rate to increase the process $\left(Y^\mathbf{N}(t)\right)$ of the {\bf total} number of infected nodes for each of the two microstate configurations in Figure~\ref{fig:contaminabipartite}. The darkened (colored) nodes represent the infected nodes.
\begin{figure} [hbt]
\begin{center}
\includegraphics[scale= 0.3]{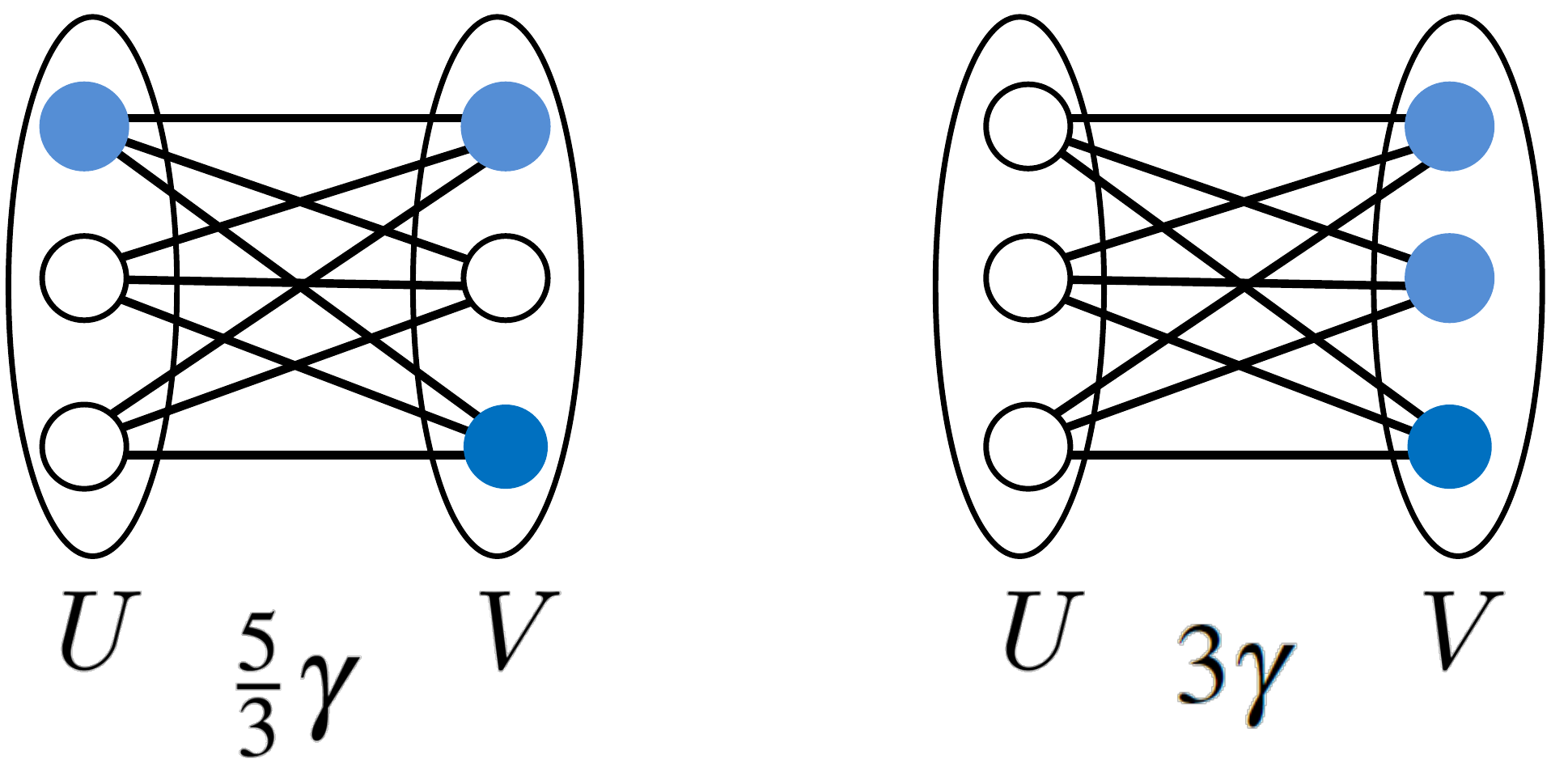}
\caption{Two microstate configurations with same number of infected nodes but different rates.}\label{fig:contaminabipartite}
\end{center}
\end{figure}
 Let $\left(Y^{\mathbf{N}}_i(t)\right)$, $i=1,2$, be the stochastic process\footnote{The components of~$\left(\mathbf{Y}^\mathbf{N}(t)\right)$ are now subindexed by the islands and not by the virus strains as in the previous subsection.} counting the number of infected individuals in each island $i\in\left\{1,2\right\}$ and $\left(\mathbf{Y}^\mathbf{N}(t)\right)=\left(Y^{\mathbf{N}}_1(t),Y^{\mathbf{N}}_2(t)\right)$. We compute the rate $\mathbf{Q}^\mathbf{N}\left(\mathbf{Y}^\mathbf{N}(t), \mathbf{Y}^\mathbf{N}(t)+\mathbf{e}_i\right)$, $i=1,2$, at which the population of infected nodes increases by one unit. For the left configuration in Figure~\ref{fig:contaminabipartite}:
\begin{align}
\mathbf{Q}^\mathbf{N}\left(\mathbf{Y}^\mathbf{N}(t), \mathbf{Y}^\mathbf{N}(t)+\mathbf{e}_1\right) {}& = \gamma Y_2^{\mathbf{N}}(t)\frac{\left(N_1-Y^{\mathbf{N}}_1(t)\right)}{N_1}=\frac{4}{3}\gamma,\label{eq:rate1}\\
\mathbf{Q}^\mathbf{N}\left(\mathbf{Y}^\mathbf{N}(t), \mathbf{Y}^\mathbf{N}(t)+\mathbf{e}_2\right) {}& = \gamma Y_1^{\mathbf{N}}(t)\frac{\left(N_2-Y^{\mathbf{N}}_2(t)\right)}{N_2}=\frac{1}{3}\gamma, \label{eq:rate2}
\end{align}
where $\mathbf{e}_i$ is the canonical vector--$i$th entry equal to~$1$ and zero at the remaining entries. The rate at which the total population of infected notes $Y^{\mathbf{N}}$ increases is the sum of these two rates:
 \begin{align*}
 \frac{5}{3}\gamma{}&=\frac{4}{3}\gamma+\frac{1}{3}\gamma.
 \end{align*}
 This is the value indicated on the left of Figure~\ref{fig:contaminabipartite}. For the configuration on the right of Figure~\ref{fig:contaminabipartite}, a similar calculation shows that the population of infected nodes increases at the rate of $3\gamma$. The two rates are different, and so, like for Example~\ref{exp:arbitrarynetwork}, $Y^{\mathbf{N}}$ is not Markov.
 \hfill$\small\blacksquare$
\end{example}
 This Example shows that the number of infected nodes $\left(Y^{\mathbf{N}}(t)\right)=\left(Y^{\mathbf{N}}_1(t)+\ldots+Y^{\mathbf{N}}_M(t)\right)$ that is a Markov process for a complete network fails to be Markov in the multipartite network case. But Example~\ref{exp:bipartitenetwork} has more structure than Example~\ref{exp:arbitrarynetwork}. The two rates in~\eqref{eq:rate1} and~\eqref{eq:rate2} do NOT depend explicitly on the microstate $\left(\mathbf{X}^\mathbf{N}(t)\right)$; they depend only on the macrostate $\left(\mathbf{Y}^\mathbf{N}(t)\right)=\left(Y^{\mathbf{N}}_1(t), Y^{\mathbf{N}}_2(t)\right)$. If we computed the rate to reduce the infected population by one (healing only at time~$t$), we would arrive at a similar conclusion--the rates depend only on the process $\left(\mathbf{Y}^\mathbf{N}(t)\right)$. That is, the rate process $\left(\mathbf{Y}^\mathbf{N}(t)\right)$ is adapted to its natural filtration, and the vector process $\left(\mathbf{Y}^\mathbf{N}(t)\right)=\left(Y_1^{\mathbf{N}}(t),Y_2^{\mathbf{N}}(t)\right)$ is now Markov. This example illustrates intuitively that we can expect to derive a low-dimensional macrostate that is Markov for the bipartite network or further multipartite networks.

%More generally, consider a single virus multipartite network with~$M$ islands, with $N_i$ being the size (number of nodes) of island~$i$.
%\begin{inparaenum}[1)]
%\item finding a Markovian macrostate for the network, i.e., global or macroscopic quantities, low dimensional functionals of the microstate $X^N(t)$ of the network, that are Markov processes;$\left(Y^{N_i}_i(t)\right)$ be the stochastic process counting the number of infected individuals in island $i$ for $i=1,\ldots,M$, and $\left(\mathbf{Y}^{\mathbf{N}}(t)\right)$ be the corresponding macrostate vector; and
%\item $\overline{Y}^{N_i}_i(t)=\frac{Y^{N_i}(t)}{N_i}$ and $\left(\overline{\mathbf{Y}}^{\mathbf{N}}(t)=\left(\overline{Y}^{\mathbf{N}}_1(t)\,\,\cdots \overline{Y}^{\mathbf{N}}_M(t)\right)\right)$ be the corresponding normalized macrostates and vector of normalized macrostates, i.e., of the fraction of infected nodes in each island.
%  \end{inparaenum}
\subsection{Mean Field Dynamics}
\label{subsec:meanfield}
Consider a single virus spread in a multipartite network with~$M$ islands, with $N_i$ being the size (number of nodes) of island~$i$. Let:
\begin{inparaenum}[1)]
\item $\left(Y^{\mathbf{N}}_i(t)\right)$ be the stochastic process counting the number of infected individuals in island $i$ for $i=1,\ldots,M$, and
     $\left(\mathbf{Y}^{\mathbf{N}}(t)\right)=\left(Y^{\mathbf{N}}_1(t)\,\,\cdots Y^{\mathbf{N}}_M(t)\right)$ be the corresponding macrostate vector; and
\item $\overline{Y}^{\mathbf{N}}_i(t)=\frac{Y^{\mathbf{N}}(t)}{N_i}$ and $\left(\overline{\mathbf{Y}}^{\mathbf{N}}(t)=\left(\overline{Y}^{\mathbf{N}}_1(t)\,\,\cdots \overline{Y}^{\mathbf{N}}_M(t)\right)\right)$ be the corresponding normalized macrostates and vector of normalized macrostates. The $M$-dimensional vectors $\left(\mathbf{Y}^{\mathbf{N}}(t)\right)$ and $\left(\mathbf{\overline{Y}}^{\mathbf{N}}(t)\right)$ collect the quantities of interest regarding the global behavior of the stochastic network. Since the number of islands in the network $M<<N$, these vectors are low dimensional, being potential candidates to be the macrostate of the network.
\end{inparaenum}

Reference~\cite{augusto_moura_emergent} shows that $\left(\mathbf{Y}^{\mathbf{N}}(t)\right)$ and $\left(\mathbf{\overline{Y}}^{\mathbf{N}}(t)\right)$ are Markov and that  $\left(\overline{\mathbf{Y}}^{\mathbf{N}}(t)\right)$ in the limit of large networks converges weakly to the solution of the following coupled differential equations, $i=1,\cdots,M,$:
\begin{align}
\label{eq:single1}
\frac{d}{dt}y_i(t) {}&=  \left(\sum_{j\sim i}\overline{\gamma}_{ji}y_j(t)\right)\left(1-y_i(t)\right)-y_i(t)
\end{align}
where the effective infection rate from island~$j$ to island~$i$ is $\overline{\gamma}_{ji}=\gamma_{ji}\times\alpha_{ji}$, with $\alpha_{ji}$ being the asymptotic size ratio between islands $j$ and $i$, i.e., $N_j/N_i\rightarrow \alpha_{ji}$. The parameter $\overline{\gamma}_{ji}$ captures the microscopic information through the rate $\gamma_{ji}$ and the relative size parameter~$\alpha_{ji}$. Without loss of generality (wlog), we take $\mu=1$. We drop the bar in the rate parameter, referring to $\overline{\gamma}_{ji}$ as $\gamma_{ji}$.

The solution to the~$M$ equations~\eqref{eq:single1} is in vector form $\mathbf{y}(t)=\left(y_1(t)\cdots y_M(t)\right)$ and the path solution is $\left(\mathbf{y}(t)\right)$. The path solution from initial condition $\mathbf{y}(0)=\mathbf{y}_0$ is $\left(\mathbf{y}\left(t,\mathbf{y}_0\right)\right)$. The function $\mathbf{y}\,:\,\left.\left[0,+\infty\right.\right)\times\left[0,1\right]^M\rightarrow \left[0,1\right]^M$ is also referred to as the flow of ODEs~\eqref{eq:single1}.

For the general bi-viral epidemics over a multipartite network, Reference~\cite{augusto_moura_emergent} further shows that the limiting dynamics for the fraction of infected nodes for the two strains of virus converges weakly, under the Skorokhod topology in the space of \emph{c\`{a}dl\`{a}g} sample paths, to the solution of the coupled vector of~$2M$ differential equations, $i=1,\cdots,M$:
\begin{align}
\frac{d}{dt}y_i(t) & = \left(\sum_{j\sim i}\overline{\gamma}^{y}_{ji} y_j(t)\right)\left(1-x_i(t)-y_i(t)\right)-y_i(t)\label{eq:multi}\\
\frac{d}{dt}x_i(t) & =\left(\sum_{j\sim i}\overline{\gamma}^{x}_{ji} x_j(t)\right)\left(1-x_i(t)-y_i(t)\right)-x_i(t),\label{eq:multi2}
\end{align}
where $y_i(t)$ and $x_i(t)$ are the limiting fractions of infected nodes by virus strains~$y$ and~$x$ in island~$i$; and $\overline{\gamma}^{y}_{ji}=\gamma^{y}_{ji}\times\alpha_{ji}$ and $\overline{\gamma}^{x}_{ji}=\gamma^{x}_{ji}\times\alpha_{ji}$ are the effective infection rates from island $j$ to island $i$. In~\eqref{eq:multi} and~\eqref{eq:multi2}, wlog~$\mu_1=\mu_2=1$. Similarly to the single virus epidemics, the path solution to~\eqref{eq:multi}-\eqref{eq:multi2}, for the two viral strains, is $\left(\mathbf{y}(t),\mathbf{x}(t)\right)$. Solutions parameterized by the initial conditions $\left(\mathbf{y}_0, \mathbf{x}_0\right)$ are represented by $\left(\mathbf{y}\left(t,\left(\mathbf{y}_0,\mathbf{x}_0\right)\right),\mathbf{x}\left(t,\left(\mathbf{y}_0,\mathbf{x}_0\right)\right)\right)$. We drop the over bars on the rates.

These limiting dynamics are derived in~\cite{augusto_moura_emergent}. The next Section studies the qualitative behavior of the~$M$ ODEs~\eqref{eq:single1} for single virus and of the~$2M$ ODEs~\eqref{eq:multi} and~\eqref{eq:multi2} for the bi-virus case.
%
%
%%%%%%%%%%%%%%%%%%%%%%%%%%%%%%%%%%%%%%%%%%%%%%%%%%%%%%%%%%%%%%%%%%%%%%%%%%%%%%%%%%%%%%%%%%%%%%%%%%%%%%%%%%%%%%%%%%%%%%%%%%%%%%%%%%%%%%%%%%%%%%%%%%%%%

%$\stackrel{(n)}{x}\!\!\!(t)$

\section{Macroscopic Model -- Bipartite Networks}\label{sec:qualitativeanalysis}
\hspace{0.43cm}
%%What this thesis is about
 We investigate the macroscopic behavior of epidemics by studying qualitatively the dynamics of the limiting vector process $\left(\mathbf{y}(t)\right)$. We build our results in steps. This section focus on \emph{bipartite} networks considering single virus epidemics in Subsection~\ref{subsec:bipartitenet-singlevirus} and multi-virus epidemics in Subsection~\ref{subsec:bipartitenet-bi-virus}; preliminary results were presented in~\cite{Paper:CDC-2012}. We then extend the analysis of single and bi-virus epidemics to \textit{regular} multipartite networks under multi-virus epidemics in Section~\ref{sec:multipartite}.
  %For regular supernetworks, we recall, the superdegree is the same for every island of the supernetwork. The components of~$\left(\mathbf{y}(t)\right)$ are the fractions of infected nodes in each island in the asymptotic limit of a large regular multipartite network. We build our results in steps. Subsection~\ref{subsec:bipartitenet-singlevirus} considers a bipartite network with single virus epidemics. Subsection~\ref{subsec:bipartitenet-bivirus} extends this study to bipartite networks and bi-virus epidemics;  Subsection~\ref{subsec:multipartitenet-singlevirus} studies arbitrary regular multipartite networks with single virus. Finally, Subsection~\ref{subsec:multipartitenet-bivirus} analyzes the epidemics of two strain virus in regular multipartite networks.

We derive conditions on the parameters of the microscopic SIS virus model and on the network structure for a macroscopic behavior to emerge in the stochastic large network--a strain perpetuates, or a survival of the fittest is observed. For a bipartite network, these questions translate into the dynamics of the density of infected nodes in each island \emph{per} strain. The mean field dynamics are characterized by a nonlinear system of coupled ODEs, see for example~\eqref{eq:single1}, that are derived in~\cite{augusto_moura_emergent} as fluid limit dynamics from the local peer-to-peer diffusion model, as we discussed in Subsection~\ref{subsec:meanfield}.

%This macroscopic behavior becomes apparent by averaging the local randomness (island-wise) of the microscopic infection model. In~\cite{augusto_moura_emergent} we achieve this through mean field analysis in the limit of large networks. The mean field dynamics are characterized by a nonlinear system of coupled ordinary differential equations, see~\eqref{eq:multi}--\eqref{eq:multi2}.
 The qualitative analysis of dynamical systems comprises characterizing their attractors and basins of attraction. In general, this is achieved by either Lyapunov theory or numerical simulations. For the coupled nonlinear equations~\eqref{eq:single1} or~\eqref{eq:multi}--\eqref{eq:multi2}, a Lyapunov function is not readily available. Instead, we explore the structure of the mean field dynamical system. We rely on the following observation that captures a special monotonous property of our system of coupled nonlinear ODEs. We state it for the single virus bipartite network and the set of ODEs~\eqref{eq:single1}.

 Consider two isomorphic copies~$B_1$ and~$B_2$ of the same bipartite network infected by the same virus. If at time~$t_1$, the bipartite network~$B_1$ presents a higher degree of infection $\mathbf{z}(t_1)=\left(z_1(t_1),z_2(t_1)\right)$ on both islands when compared to the infection level  $\mathbf{y}(t_1)=\left(y_1(t_1),y_2(t_1)\right)$ in~$B_2$, then the epidemics state of~$B_1$ will dominate the state of~$B_2$ for all future times, i.e., $\mathbf{z}(t)\geq \mathbf{y}(t)$ for all $t\geq t_1$. In particular, if the initial states of islands~$1$ and~$2$ are given by $\mathbf{z}(0)=\left(z_1(0),z_2(0)\right)$ and $\mathbf{y}(0)=\left(y_1(0),y_2(0)\right)$ with $z_1(0)\geq y_1(0)$ and $z_2(0)\geq y_2(0)$ then, the infection rate $\left(\mathbf{y}(t)\right)$ for~$B_2$ is upperbounded by the infection rate $\left(\mathbf{z}(t)\right)$ for~$B_1$ for all $t\geq 0$. More generally, this property holds for regular multipartite networks and will be particularly explored to establish survival of the fittest: at most the strongest strain persists in the network and the remaining weaker ones necessarily die out. This turns out to be a crucial observation since for the symmetric bipartite network ($N_1=N_2$) with symmetric initial conditions, $y_1(0)=y_2(0)=y_0$, the induced solution $\left(\mathbf{y}\left(t,\mathbf{y}_0\right)\right)= \left(y\left(t,\left(y_0,y_0\right)\right), y\left(t,\left(y_0,y_0\right)\right)\right)$ can be easily characterized, and it can be used to bound the solutions of more general infection regimens.
 %Namely, one can tightly bound the solution of a generic system by resorting to a subfamily of simpler solutions.

\textbf{Preliminary Notation.}
%\hspace{0.43cm}
We summarize the main notation used throughout this section: $\stackrel{(n)}{y}_{\!\!i}\!\!(t)$:
 $n$th derivative of the fraction of $y$-infected nodes at island~$i$ at time~$t$, $t\geq 0$;
 $\mathcal{N}(i)$: represents the $1$-hop neighborhood of island~$i$;
 $\mathcal{N}^{2}(i)$: represents the $2$nd order neighborhood of~$i$, that is, $j\in \mathcal{N}^{2}(i)$ if and only if the shortest path connecting~$i$ and~$j$ (a.k.a.~geodesic) has a length of $2$ hops;
 $j\in\mathcal{N}^{n}(i)$: if and only if there exists $k\in\mathcal{N}^{n-1}(i)$ with $j\sim k$, i.e., the geodesic connecting $i$ and $j$ comprises $n$ hops;
 $\mathbf{x}\leq \mathbf{y}\in\mathbb{R}^{n}$: means $\mathbf{y}-\mathbf{x}\in\mathbb{R}^n_{+}$;
 $x\wedge y$: equal to $x\in\mathbb{R}$ if $x<y$ or equal to $y\in\mathbb{R}$, otherwise;
 $\phi\left(t,\phi_0\right)$:  represents the solution of ordinary differential equation $\dot{y}=F(y)$ as a flow $\phi\,:\,\mathbb{R}_{+}\times D\rightarrow D$, representing the state of the system at time $t$ with initial state $\phi_0$;
 $\mathbf{1}_n\in\mathbb{R}^n$: vector with all entries equal to one. The  subindex may be omitted whenever there is no room for ambiguity; and
 $\Delta_n$: simplex in $\mathbb{R}^n$ defined as $\Delta_n=\left\{\mathbf{v}\in \mathbb{R}^n_{+}\,:\,\left\langle \mathbf{v},\mathbf{1}_n\right\rangle\leq 1\right\}$, where $\langle\cdot,\cdot\rangle$ is the standard Euclidean inner product of vectors.
\subsection{Bipartite network: Single virus}\label{subsec:bipartitenet-singlevirus}
This Section considers single virus epidemics in a bipartite network. A graph is bipartite when the number of islands in the multipartite network is two. We first define a bipartite symmetric configuration that will be explored through the rest of this section.

 \begin{definition}[Symmetric configuration]\label{def:symmetricconfiguration}
 The single virus epidemics in a bipartite network has a symmetric configuration if and only if:
 \begin{inparaenum}[1)]
 \item Symmetric network, $\alpha_{12}=\alpha_{21}$ (islands have asymptotically the same size,) that is, $\gamma_{12}=\gamma_{21}$;
 \item Normalized healing rate $\mu=1$.
 \end{inparaenum}
 \hfill$\small\blacksquare$
 \end{definition}
% This Definition is generalized in a straightforward way to multipartite networks and multi-virus epidemics.

 We rewrite~\eqref{eq:single1} for the symmetric configuration. The limiting rates of infection $\left(y_1(t)\right)$ and $\left(y_2(t)\right)$ of occupancy in islands $1$ and $2$, respectively, are given by:
\begin{eqnarray}
\label{eq:bipartite1}
\frac{d}{dt}y_1(t)& = &\gamma y_2(t)\left(1-y_1(t)\right)-y_1(t)\\
\label{eq:bipartite1b}
\frac{d}{dt}y_2(t) & = &\gamma y_1(t)\left(1-y_2(t)\right)-y_2(t).
\end{eqnarray}
The solution $\left(\mathbf{y}(t)\right)=\left(y_1(t),y_2(t)\right)$ to~\eqref{eq:bipartite1}-\eqref{eq:bipartite1b} with initial condition $\mathbf{y}_0\in \left[0,1\right]^2$ exists and is unique since the dynamics are (globally) Lipschitz over the domain $D=\left[0,1\right]\times \left[0,1\right]$. Note that the set~$D$ is invariant with respect to the dynamics, that is, if $\mathbf{y}(0)=\left(y_1(0),y_2(0)\right)\in D$ then, $\mathbf{y}(t,y(0))\in D$, $\forall{t\geq 0}$. This follows of course from the underlying physical system, and it is easily established from the (ODE) limiting dynamics. The fact that $D$ is compact further implies that the solutions are defined for all $t$, $t\geq 0$.

 We determine the qualitative behavior of the coupled system of two nonlinear ODEs~\eqref{eq:bipartite1}-\eqref{eq:bipartite1b}, i.e., their critical points and corresponding basins of attraction. There are two critical points: $y^{\left(\mbox{\scriptsize eq$_1$}\right)}=\left(1-\frac{1}{\gamma}, 1-\frac{1}{\gamma} \right)$ and $y^{\left(\mbox{\scriptsize eq$_2$}\right)}=0$. We will show $y^{\left(\mbox{\scriptsize eq$_1$}\right)}$ is a global attractor if $\gamma>1$, otherwise $y_{1,2}(t)\rightarrow 0$. In words, the $y$-virus survives if $\gamma>1$, otherwise, it eventually dies out.

 The next Theorem reveals a monotone aspect of the dynamical system~\eqref{eq:bipartite1}-\eqref{eq:bipartite1b} that will be further explored in a more general setting -- an upper-bound on the initial conditions is preserved by the flow of the dynamical system~\eqref{eq:bipartite1}-\eqref{eq:bipartite1b} through all time $t\geq 0$.
\begin{theorem}
\label{th:invariant1}
Let $\left(\mathbf{y}\left(t,\mathbf{y}(0)\right)\right)_{t\geq 0}$ be the solution of~\eqref{eq:bipartite1}-\eqref{eq:bipartite1b} with initial condition $\mathbf{y}(0)\in D$. Then,
\begin{align*}
\mathbf{y}(0)\leq \mathbf{y}_0\in D \Rightarrow \mathbf{y}(t,\mathbf{y}(0))\leq \mathbf{y}(t,\mathbf{y}_0),\,\,\:\:\forall\,\,\,{t\geq 0}.
%\hspace{3cm}
\blacksquare
\end{align*}
%\hfill$\small \blacksquare$
\end{theorem}

\begin{proof}
If $\mathbf{y}(0)= \mathbf{y}_0$ then, by uniqueness $\mathbf{y}(t,\mathbf{y}(0))=\mathbf{y}(t,\mathbf{y}_0)$ for all $t\geq 0$, and the result holds. Now, let $\mathbf{y}(0)\leq \mathbf{y}_0$ with $\mathbf{y}(0)\neq \mathbf{y}_0$. Define $T=\inf\left\{t\,:\,t\geq 0\,,  \mathbf{y}(t,\mathbf{y}(0))\nleq \mathbf{y}\left(t,\mathbf{y}_0\right)\right\}$ and assume that $T< +\infty$. Since the flow is continuous and uniqueness is preserved for all $t\geq 0$, then, $y_1(T,\mathbf{y}(0))=y_1\left(T,\mathbf{y}_0\right)$ and $y_2(T,\mathbf{y}(0))<y_2\left(T,\mathbf{y}_0\right)$  (up to a relabeling.) Observe from equations~\eqref{eq:bipartite1}-\eqref{eq:bipartite1b} that $\dot{y}_1(T,\mathbf{y}(0))<\dot{y}_1\left(T,\mathbf{y}_0\right)$. Therefore,
\begin{equation}
\exists\,\,{\epsilon_1>0}\,:\,y_1(t,\mathbf{y}(0))<y_1\left(t,\mathbf{y}_0\right),\,\,\,\: \forall\,\,\,{T<t<T+\epsilon_1}.\nonumber
\end{equation}
Moreover,
\begin{align*}
y_2(T,\mathbf{y}(0))<y_2\left(T,\mathbf{y}_0\right)\Rightarrow \exists\,\,{\epsilon_2>0}:
y_2(t,\mathbf{y}(0))<y_2\left(t,\mathbf{y}_0\right),\, \forall\,\,\,{T<t<T+\epsilon_2}.
\end{align*}
Thus, we conclude that $\mathbf{y}(T+\epsilon,\mathbf{y}(0))\leq \mathbf{y}\left(T+\epsilon,\mathbf{y}_0\right)$, where $\epsilon=\epsilon_1\wedge \epsilon_2$. This contradicts the definition of $T$ and the assumption that it is finite.
%\hfill$\small \blacksquare$
\end{proof}

Before completing the analysis for the bipartite single virus case, we consider the simple case where the initial infection rates are the same, i.e., $y_1(0)=y_2(0)=y(0)=y_0$. Then, we claim, $y_1\left(t,\left(y_0,y_0\right)\right)=y_2\left(t,\left(y_0,y_0\right)\right),\,\forall\,\,\,{t\geq 0}$. Indeed, if $\left(z(t)\right)$ is solution of
\begin{equation}
\label{eq:simplificado}
\frac{d}{dt}z(t)\left[\begin{array}{c} 1\\ 1\end{array}\right]=\gamma z(t)\left(\left(1-z(t)\right)-z(t)\right)\left[\begin{array}{c} 1\\ 1\end{array}\right],
\end{equation}
it is easy to check that $z(t)\rightarrow 0$ if $\gamma\leq1$, and that $z(t)\rightarrow 1-\frac{1}{\gamma}$ if $\gamma>1$, regardless of the initial conditions.

The next Theorem builds on Theorem~\ref{th:invariant1} to complete the analysis for the bipartite single virus case, namely, it implies that, if $\gamma>1$, then the virus survives, otherwise, it dies out.
\begin{theorem}
\label{th:attractor}
Let $\left(\mathbf{y}\left(t,\mathbf{y}_0\right)\right)$ be the solution of~\eqref{eq:bipartite1}-\eqref{eq:bipartite1b} with $\mathbf{y}_0\neq \mathbf{0}$. Then,
\begin{eqnarray}
\gamma>1 & \Rightarrow & \mathbf{y}(t) \rightarrow \left(1-\frac{1}{\gamma}, 1-\frac{1}{\gamma}\right)\nonumber \\
\gamma\leq1 & \Rightarrow & \mathbf{y}(t)\rightarrow 0.\nonumber
\hspace{3cm}\blacksquare
\end{eqnarray}
%\hfill$\small \blacksquare$
\end{theorem}
\begin{proof}
First, assume $\mathbf{y}(0)=\mathbf{y}_0>\mathbf{0}$ and $\gamma>1$. Choose $\epsilon>0$ so that $\mathbf{y}_0>\epsilon \mathbf{1}_2>0$. From Theorem~\ref{th:invariant1}, $\mathbf{y}\left(t,\mathbf{y}_0\right)\geq \mathbf{y}\left(t,\epsilon \mathbf{1}_2\right)$, $\forall\,\,{t\geq 0}$. Thus,
\begin{equation}
\lim_{t\rightarrow\infty}\inf \mathbf{y}\left(t,\mathbf{y}_0\right)\geq
\lim_{t\rightarrow\infty} \mathbf{y}\left(t,\epsilon \mathbf{1}_2\right)=\left(1-\frac{1}{\gamma}\right)\mathbf{1}_2.\nonumber
\end{equation}
The last equality follows from the asymptotics of~\eqref{eq:simplificado}. Similarly, we upperbound the solution by $\mathbf{y}\left(t,\mathbf{y}_0\right)\leq \mathbf{y}\left(t,\mathbf{1}_2\right)$, $\forall\,\,{t\geq 0}$. Thus
\begin{equation}
\lim_{t\rightarrow\infty}\sup \mathbf{y}\left(t,\mathbf{y}_0\right)\leq \left(1-\frac{1}{\gamma}\right)\mathbf{1}_2.\nonumber
\end{equation}
Now, assume $y_1(0)=0$ and $y_2(0)>0$. Then, $\dot{y}_1(0)=\gamma y_2(0)>0$. Therefore, by the same argument as in the proof of Theorem~\ref{th:invariant1}, there exists $T>0$ so that $\left(\mathbf{y}\left(t,\mathbf{y}_0\right)\right)>0$, $\forall\,\,\,{0<t<T}$. Choose, $t_0\in\left(0,T\right)$. Then, $\mathbf{y}\left(t,\mathbf{y}_0\right)=\mathbf{y}\left(t-t_0,\mathbf{y}\left(t_0\right)\right)$, $\forall\,\,{t\geq t_0}$. Since $\mathbf{y}\left(t_0\right)>0$,
\begin{equation}
\lim_{t\rightarrow\infty} \mathbf{y}\left(t,\mathbf{y}_0\right)=\lim_{t\rightarrow\infty} \mathbf{y}\left(t-t_0,\mathbf{y}\left(t_0\right)\right)=\left(1-\frac{1}{\gamma}\right)\mathbf{1}_2.\nonumber
\end{equation}
The argument repeats for $\gamma\leq 1$.

Alternatively, we can prove Theorem~\ref{th:attractor} by defining the error function
\begin{equation}
w\left(\mathbf{y}\right):=\frac{1}{2}\left(y_1-y_2\right)^2\geq 0,
\end{equation}
which, for all time $t$, $t\geq 0$ and any solution $\left(\mathbf{y}(t)\right)$ of~\eqref{eq:bipartite1}-\eqref{eq:bipartite1b}, leads to
\begin{eqnarray*}
\frac{d}{dt}w\left(\mathbf{y}(t)\right) & = &\phantom{-}\left(y_1(t)-y_2(t)\right)\left(\gamma \left(y_2(t)-y_1(t)\right)-
\left(y_1(t)-y_2(t)\right)\right)\\
 & = & -\left(y_1(t)-y_2(t)\right)^2\left(\gamma+1\right)\leq 0,
\end{eqnarray*}
 In words, $w$ is a Lyapunov function for the attractor given by the set of configurations where islands are evenly infected, i.e., the straight line $r=\left\{\left(y_1,y_2\right)\in \mathbb{R}^2\,:\,y_1=y_2\right\}\cap D$. Since the set $D$ is compact and the singleton $\left\{\left(1-\frac{1}{\gamma},1-\frac{1}{\gamma}\right)\right\}$ is the maximally invariant subset of the straight line $r$ for the dynamics~\eqref{eq:bipartite1}-\eqref{eq:bipartite1b}, Theorem~\ref{th:attractor} follows.
%\hfill$\small \blacksquare$
\end{proof}
It is not clear how to extend this alternative proof to Theorem~\ref{th:attractor} to the more general cases of two virus or over multipartite networks. Therefore, in the following Sections, we explore the monotonicity property of the dynamical system to analyze these more general cases, starting in the next Subsection, by extending the analysis to bi-viral infection in bipartite networks.
\subsection{Bipartite Network: Bi-viral Epidemics}\label{subsec:bipartitenet-bi-virus}
Consider two viruses~$x$ and~$y$, and $x_i(t),\,y_i(t)$ be the fractions of $x$- and $y$-infected nodes at island $i$, $i=1,2$, at time $t\geq 0$ for the limiting dynamics. We consider the symmetric configuration in Definition~\ref{def:symmetricconfiguration} with micro infection parameters $\gamma^x$ and~$\gamma^y$ for the virus~$x$ and~$y$. We write~\eqref{eq:multi}-\eqref{eq:multi2} for the bi-virus epidemics for a bipartite symmetric configuration, $i,j=1,2, \:i\neq j$:
\begin{eqnarray}
\frac{d}{dt}y_i(t) \!\!\!\!& = &\!\!\!\! \gamma^y y_j(t) \left(1-y_i(t)-x_i(t)\right)-y_i(t)\label{eq:bipartite2}\\
\frac{d}{dt}x_i(t) \!\!\!\!& = &\!\!\!\!\gamma^x x_j(t)\left(1-y_i(t)-x_i(t)\right)-x_i(t)\label{eq:bipartite22}.
\end{eqnarray}
%with $i,j=1,2$ and $i\neq j$, where we assume again islands with the same size, that is, $\gamma^y=\gamma^y_{12}=\gamma^y_{21}$ and $\gamma^x=\gamma^x_{12}=\gamma^x_{21}$.
 From the exclusion principle, the sets of $x$- and $y$-infected nodes in island $i$ are disjoint with $0 \leq y_i(t)+x_i(t)\leq 1$, $\forall t\geq 0$ and $i=1,2$. The invariant domain is $\widehat{D}=\Delta_2\times\Delta_2$,
 \begin{equation}
 \Delta_2=\left\{\left(x,y\right)\in\mathbb{R}^2:x,y\geq 0\, \wedge\, x+y\leq1\right\},\nonumber
 \end{equation}
the simplex in $\mathbb{R}^{2}$. In Subsection~\ref{subsec:bipartitenet-singlevirus}, the solutions symmetrically initialized -- namely, $y_1(0)=y_2(0)$ -- were easily characterized, and any solution was appropriately lower/upper bounded by such \emph{easy} solutions, from which we determined the long term behavior of any solution. We extend this to bi-virus. We start by extending Theorem~\ref{th:invariant1}.
\begin{theorem}
\label{th:invariant2}
Let $\left(\mathbf{z}\left(t,\mathbf{z}(0)\right)\right)_{t\geq 0}$ solve~\eqref{eq:bipartite2}-\eqref{eq:bipartite22} where $\mathbf{z}(0)=\left(\mathbf{x}(0),\mathbf{y}(0)\right)\in\widehat{D}$, $\mathbf{y}(0)=\left(y_1(0),y_2(0)\right)$, $\mathbf{x}(0)=\left(x_1(0),x_2(0)\right)$.
Then, $\forall\,\,{t\geq 0}$, $\mathbf{x}_0=\left(x_{01},x_{02}\right)$, $\mathbf{y}_0=\left(y_{01},y_{02}\right)\in\Delta_2$, $\mathbf{x}(t)=\left(x_1(t),x_2(t)\right)$, $\mathbf{y}(t)=\left(y_1(t),y_2(t)\right)$, and $\mathbf{z_0}=\left(\mathbf{x}_0,\mathbf{y}_0\right)$:
\begin{align*}
\mathbf{y}(0)\leq \mathbf{y}_0 \mbox{  and  } \mathbf{x}(0)\geq \mathbf{x}_0
 \Rightarrow \mathbf{y}(t,\mathbf{z}(0))\leq \mathbf{y}\left(t,\mathbf{z}_0\right)\mbox{  and  }\mathbf{x}\left(t,\mathbf{z}(0)\right)\geq \mathbf{x}\left(t,\mathbf{z}_0\right).
 %\hspace{3cm}
 \blacksquare
\end{align*}
%\hfill$\small \blacksquare$
\end{theorem}

\begin{proof}
Of course, if $\mathbf{y}(0)=\mathbf{y}_0$ and $\mathbf{x}(0)=\mathbf{x}_0$ (i.e., $\mathbf{z}(0)=\mathbf{z}_0$), then, by uniqueness, $\mathbf{z}\left(t,\mathbf{z}(0)\right)=\mathbf{z}\left(t,\mathbf{z}_0\right)$, $\forall\,\,{t\geq 0}$, and the Theorem holds. Let us further assume that $\mathbf{y}(0)\neq \mathbf{y}_0$. Similarly to the proof of Theorem $\ref{th:invariant1}$, define
\begin{equation}
T=\inf\left\{t\,:\,t\geq 0,\,  \mathbf{y}\left(t,\mathbf{y}(0)\right)\nleq \mathbf{y}\left(t,\mathbf{y}_0\right)\mbox{  or  } \mathbf{x}\left(t,\mathbf{x}(0)\right)\ngeq \mathbf{x}\left(t,\mathbf{x}_0\right)\right\}.\nonumber
\end{equation}

Assume $T<\infty$. Then, with $i\neq j$, $i,j\in\left\{1,2\right\}$, we have one of the two configurations below:
\begin{align}
\label{eq:first}
y_i\left(T,\mathbf{z}(0)\right)&=y_i\left(T,\mathbf{z}_0\right)\mbox{  and  }y_j\left(T,\mathbf{z}(0)\right)<y_j\left(T,\mathbf{z}_0\right)
\\
\label{eq:firstb}
x_i\left(T,\mathbf{z}(0)\right)&=x_i\left(T,\mathbf{z}_0\right)\mbox{  and  }x_j\left(T,\mathbf{z}(0)\right)>x_j\left(T,\mathbf{z}_0\right).
\end{align}
Without loss of generality, choose configuration~\eqref{eq:first} with $i=1$ and $j=2$.

\textit{Case~1}: If $x_1\left(T, \mathbf{z}(0)\right)+y_1\left(T,\mathbf{z}(0)\right)<1$, then, from \eqref{eq:bipartite2} and \eqref{eq:bipartite22}, we have
\begin{equation}
\dot{y}_1\left(T,\mathbf{z}(0)\right)<\dot{y}_1\left(T,\mathbf{z}_0\right).\nonumber
\end{equation}
Therefore,
\begin{equation}
\exists\,\,\epsilon_1>0:\,y_1\left(t,\mathbf{z}(0)\right)<y_1\left(t,\mathbf{z}_0\right),\,\,\,\:\forall\,\,\, T<t<T+\epsilon_1.\nonumber
\end{equation}
Also,
\begin{equation}
y_2\left(T,\mathbf{z}(0)\right)<y_2\left(T,\mathbf{z}_0\right)\Rightarrow \exists\,\,\epsilon_2>0:\,y_2\left(t,\mathbf{z}(0)\right)<y_2\left(t,\mathbf{z}_0\right),\,\,\,\:\forall\,\,\, T<t<T+\epsilon_2.\nonumber
\end{equation}
Thus,
\begin{equation}
y\left(t,\mathbf{z}(0)\right)\leq y\left(t,\mathbf{z}_0\right),\,\,\,\forall\,\,\,\: T<t<T+\epsilon\nonumber
\end{equation}
with $\epsilon=\epsilon_1\wedge \epsilon_2$. In the same way, we can conclude that for some $\alpha>0$:
\begin{equation}
x\left(t,\mathbf{z}(0)\right)\geq x\left(t,\mathbf{z}_0\right),\,\,\,\:\forall\,\,\, T<t<T+\alpha.\nonumber
\end{equation}

\textit{Case~2}: If $x_1\left(T,\mathbf{z}(0)\right)+y_1\left(T,\mathbf{z}(0)\right)=1$, then, for all $t\in\left(T,T+\epsilon\right)$:
\begin{equation}
\dot{x}_1\left(T,\mathbf{z}(0)\right)+\dot{y}_1\left(T,\mathbf{z}(0)\right)=-\left(x_1(T)+y_1(T)\right)<0\Rightarrow\exists\,\, \epsilon>0:\,x_1\left(t,\mathbf{z}(0)\right)+y_1\left(t,\mathbf{z}(0)\right)<1. \nonumber
\end{equation}
 From case $1$, we reach a contradiction on $T$, and the Theorem is proved.
%\hfill$\small \blacksquare$
\end{proof}
Figure~\ref{fig:invariants_main} depicts geometrically Theorem~\ref{th:invariant2} as the monotonous property in Theorem~\ref{th:invariant2} is equivalent to the invariance of the set given by the Cartesian product of the dark (colored) triangles in Figure~\ref{fig:invariants_main}.
\begin{figure} [hbt]
\begin{center}
\includegraphics[scale= 0.3]{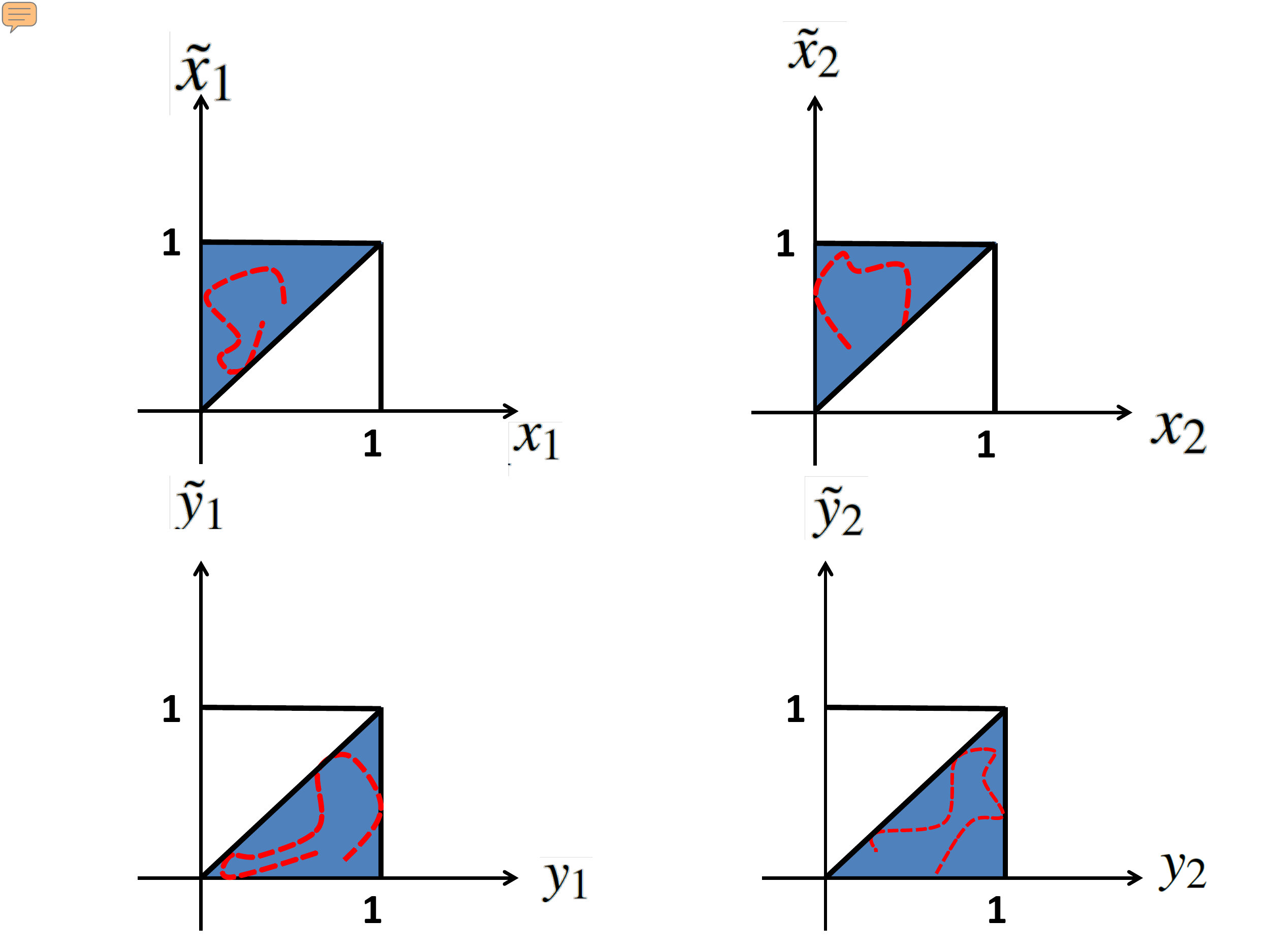}
\caption{Phase space of the augmented dynamical system $\left(\widetilde{x}_1(t),x_1(t),\widetilde{y}_1(t),y_1(t), \widetilde{x}_2(t),x_2(t),\widetilde{y}_2(t),y_2(t)\right)$, i.e., the  evolution of the viral evolution in two isomorphic bipartite networks with perhaps different initial degrees of infection. The red (dashed) curve captures the idea that a solution cannot escape the blue (dark) region in finite time.}\label{fig:invariants_main}
\end{center}
\end{figure}

Similarly to as done in the previous subsection, given any initial condition
\begin{equation}
\left(x_1(0),x_2(0),y_1(0),y_2(0)\right)>0,
\end{equation}
we may choose $\widetilde{x}_1(0)=\widetilde{x}_2(0)=\max\{x_1(0),x_2(0)\}$ and
$\widetilde{y}_1(0)=\widetilde{y}_2(0)=\min\{y_1(0),y_2(0)\}$. In this case, $\left(\widetilde{x}_1(t),\widetilde{x}_2(t),\widetilde{y}_1(t),\widetilde{y}_2(t)\right)$ is solution of the reduced system
\begin{eqnarray}
\dot{\widetilde{x}}(t){\bf 1}_2 &\!\!\! =\!\!\! & \left(\gamma^x \widetilde{x}(t)\left(1-\widetilde{x}(t)-\widetilde{y}(t)\right)-\widetilde{x}(t)\right){\bf 1}_2\label{eq:reduced}\\
\dot{\widetilde{y}}(t){\bf 1}_2 & \!\!\!=\!\!\! & \left(\gamma^y \widetilde{y}(t)\left(1-\widetilde{x}(t)-\widetilde{y}(t)\right)-\widetilde{y}(t)\right){\bf 1}_2\label{eq:reduced2}.
\end{eqnarray}
Remark that these equations~(\ref{eq:reduced}) and~(\ref{eq:reduced2}) represent the dynamics of bi-viral epidemics over a complete network as studied in~\cite{paper:CDC}. Therefore, if $\gamma^x>\gamma^y$
\begin{eqnarray}
\lim_{t\rightarrow \infty}\widetilde{x}_1(t) &=&\left(1-\frac{1}{\gamma^x}\right)\nonumber\\
\lim_{t\rightarrow \infty}\widetilde{x}_2(t) &= &\left(1-\frac{1}{\gamma^x}\right)\nonumber\\
\lim_{t\rightarrow \infty}\widetilde{y}_1(t) &=& 0\nonumber\\
\lim_{t\rightarrow \infty}\widetilde{y}_2(t) &=& 0,\nonumber
\end{eqnarray}
regardless of the initial conditions. Also, choosing $\overline{x}_1(0)=\overline{x}_2(0)=\min\{x_1(0),x_2(0)\}$ and $\overline{y}_1(0)=\overline{y}_2(0)=\max\{y_1(0),y_2(0)\}$, we have
\begin{eqnarray}
\lim_{t\rightarrow \infty}\overline{x}_1(t) &=&\left(1-\frac{1}{\gamma^x}\right)\nonumber\\
\lim_{t\rightarrow \infty}\overline{x}_2(t) &= &\left(1-\frac{1}{\gamma^x}\right)\nonumber\\
\lim_{t\rightarrow \infty}\overline{y}_1(t) &=& 0\nonumber\\
\lim_{t\rightarrow \infty}\overline{y}_2(t) &=& 0.\nonumber
\end{eqnarray}
Therefore, since (from Theorem~\ref{th:invariant2}) $\overline{x}_i(t)\leq x_i(t)\leq \widetilde{x}_i(t)$ and $\overline{y}_i(t)\leq y_i(t)\leq \widetilde{y}_i(t)$, $\forall\,_{t\geq 0}$, then
\begin{eqnarray}
\lim_{t\rightarrow \infty}x_1(t) &=&\left(1-\frac{1}{\gamma^x}\right)\nonumber\\
\lim_{t\rightarrow \infty}x_2(t) &= &\left(1-\frac{1}{\gamma^x}\right)\nonumber\\
\lim_{t\rightarrow \infty}y_1(t) &=& 0\nonumber\\
\lim_{t\rightarrow \infty}y_2(t) &=& 0.\nonumber
\end{eqnarray}
The case when $x_i(0)=0$ or $y_i(0)=0$ for some $i=1,2$ is treated similarly as in the proof of Theorem~\ref{th:invariant1}, in that there exists $\delta>0$ so that $x_i(t)>0$ for all $t\in \left(0,\delta\right)$ and $i=1,2$ as long as $x_1(0)>0$ or $x_2(0)>0$. Otherwise, $\left(x_1(0),x_2(0)\right)=0$ is an equilibrium point (no virus of type $x$ in the system) and $\left(x_1(t),x_2(t)\right)=0$ for all $t\geq 0$.
Figure~\ref{fig:bound} illustrates the possibility of bounding any configuration by simpler symmetric well-characterized configurations. Such bounds are preserved for all $t$, $t\geq 0$ as established in Theorem~\ref{th:invariant2}.
\begin{figure} [hbt]
\begin{center}
\includegraphics[scale= 0.3]{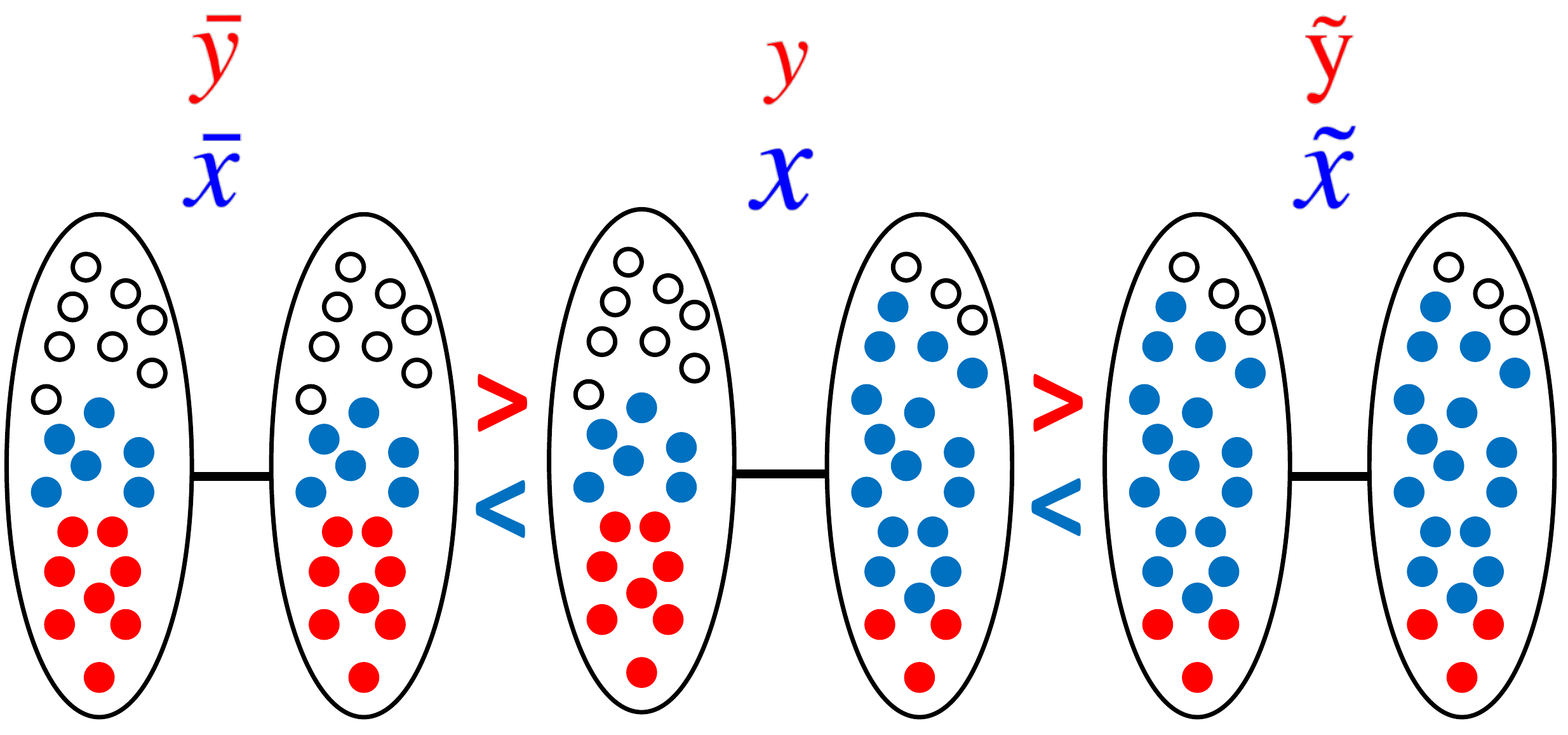}
\caption{Population of blue (lighter color) in the center bipartite network is lower and upper bounded by the corresponding populations in the left and right bipartite networks. The same goes, in the other way around, for the red (darker) population. The symmetric configurations in the left and right bipartite networks induce well-known solutions that force to the same equilibrium state the configuration over time of the middle bipartite network.}\label{fig:bound}
\end{center}
\end{figure}

\section{Macroscopic Behavior -- Regular Multipartite Networks}
\label{sec:multipartite}
This Section extends the results on the macroscopic behavior for bipartite networks in Section~\ref{sec:qualitativeanalysis} to arbitrary \emph{regular} multipartite networks. We recall that a multipartite network is regular if the superdegree is the same for every island in the supernetwork. Subsection~\ref{subsec:multipartitenet-singlevirus} considers single virus infection, while Subsection~\ref{subsec:multipartitenet-bi-virus} analyzes the epidemics of multiple virus strains. The focus is again on the qualitative dynamics of the vector process $\left(\mathbf{y}(t)\right)$ of the fractions of infected nodes in each island by each virus in the asymptotic limit of large multipartite networks.

\subsection{Regular Multipartite Network: Single Virus}\label{subsec:multipartitenet-singlevirus}
%\hspace{0.43cm}
%
We study a single virus in a \textit{regular} multipartite network. Wlog, we consider the symmetric configuration in Definition~\ref{def:symmetricconfiguration} where all islands of the multipartite network have the same size and the inter rates of infection are equal, $\gamma_{ij}\equiv \gamma$, $1\leq i,j\leq M$, $i\neq j$. In a multipartite network, whenever a node from one island connects to a node from another island, then any node from the first island connects to any node in the second island. The mean field dynamics of a single virus epidemics over a large \textit{symmetric configuration} \textit{regular} multipartite network with~$M$ islands is obtained by specializing~\eqref{eq:single1} to the symmetric configuration. We get the~$M$ coupled nonlinear ordinary differential equations, $i=1,\cdots, M$:
\begin{equation}
\label{eq:multipartite}
\frac{d}{dt}y_i(t)=\underbrace{\left(\gamma\sum_{j\sim i}y_j(t)\right)\left(1-y_i(t)\right)-y_i(t)}_{F_i(y)}.
\end{equation}
%for $i=1,\ldots,M$ where $M$ represents the number of islands and we define the vector field $F(y)=\left(F_1(y),\ldots,F_M(y)\right)$. Again, we assume same sized islands, that is, $\gamma=\gamma_{ij}$ for all $i,j\in \left\{0,1,\ldots,M\right\}$.
  We define the vector field $\mathbf{F}(\mathbf{y})=\left(F_1(\mathbf{y}),\ldots,F_M(\mathbf{y})\right)$. The next two Theorems are crucial to establishing the main result of this subsection in Theorem~\ref{th:convergence} and, moreover, they reveal the qualitative impact of the super-topology on the regularity of the solutions. Namely, they state that the degree of infection at island $j$ has an impact on the fraction of infected nodes at island $i$, $n$-hops away from $j$, through perturbations of its $n$th- (or higher than~$n$) order derivative.
%\stackrel{(n)}{y}_{\!\!i}\!\!(t)
%
\begin{theorem}
\label{th:ordem}
Let $y_i(t)=0$ and $y_j(t)=0$, $\forall\,\,\,{j\in\cup_{l=1}^{n}\mathcal{N}^l(i)}$ for some time $t\geq0$. Then,
\begin{equation}
\label{eq:order2}
\stackrel{(\ell)}{y}_{\!\!i}\!\!(t)=0,\,\,\,\forall\,\,\,{0<\ell\leq n},
\end{equation}
that is, if there are no infected islands within a neighborhood up to order~$n$ of island~$i$, then all derivatives of $y_i(t)$ up to order~$n$ are zero.
\hfill$\small \blacksquare$
\end{theorem}

\begin{proof}
We apply induction on the order $n$.

\textit{Step~1}: For $n=1$:
\begin{equation}
\dot{y}_i(t)=\left(\gamma \sum_{q\sim i}y_q(t)\right)\left(1-y_i(t)\right)-y_i(t)=0.\nonumber
\end{equation}

\textit{Step~2}: \textbf{Induction step}. We assume that Theorem~\ref{th:ordem} holds for $n-1$ and prove it holds for order~$n$. By algebraic and reordering manipulations, we can show:
\begin{equation}
\stackrel{(n)}{y}_{\!\!i}\!\!(t)=\underbrace{\left(\gamma\sum_{k\sim i}\stackrel{(n-1)}{y}_{\!\!\!\!\!\!k}\!\!(t) \right)}_{\bf A} -\underbrace{\sum_{\ell=0}^{n-1}\left(\begin{array} {c} n-1\\\ell\end{array}\right)\stackrel{(\ell)}{y}_{\!\!i}\!\!(t)\left(\gamma\sum_{q\sim i} \stackrel{(n-1-\ell)}{y}_{\!\!\!\!\!\!\!\!q}\!\!(t)\right)}_{\bf B}- \underbrace{\stackrel{(n-1)}{y}_{\!\!\!\!\!\!i}\!\!(t)}_{\bf C}\nonumber
\end{equation}
holds for all $n\in\mathbb{N}$. We analyze now each term. First, note that from the induction hypothesis $\stackrel{(\ell)}{y}_{\!\!i}\!\!(t)=0$ for all $l=1,\ldots,n-1$.

${\bf A}$: Since by assumption $y_j(t)=0$, $\forall\,\,\,{j\in\cup_{l=1}^{n} \mathcal{N}^{l}}(i)$, then, if $k\in\mathcal{N}(i)$, by induction, $\stackrel{(\ell)}{y}_{\!\!k}\!\!(t)=0$ for all $l=1,\ldots,n-1$. Therefore, $\gamma\sum_{k\sim j}\stackrel{(n-1)}{y}_{\!\!\!\!\!\!k}\!\!(t)=0$, i.e., term~${\bf A}$ is zero.

${\bf B}$: By assumption, $y_i(t)=0$, and by induction, for all $l=1,\ldots,n-1$, $\stackrel{(\ell)}{y}_{\!\!i}\!\!(t)=0$, hence term~${\bf B}$ is zero.

${\bf C}$: Term ${\bf C}$ is zero, since by induction $\stackrel{(n-1)}{y}_{\!\!\!\!\!\!i}\!\!(t)=0$.
%\hfill$\small \blacksquare$
\end{proof}
%The previous Theorem stated that the low-order moments (or derivatives) of an infected population at island $i$ are not sensitive to further away infected islands.

The next Theorem states that higher order moments are sensitive to further away infected islands--island $i$ located $n$-hops away from island $j$, affects only the $n$th-order derivative of~$j$.
\begin{theorem}
\label{th:order2}
Let $y_i(t)>0$ and $y_j(t)=0$, $\forall\,\,\,{j\neq i}$ for some time $t\geq0$. Then, $\forall {\ell<n}$
\begin{equation}
\label{eq:order22}
%\hspace{-.2cm}
j\!\in\!\mathcal{N}^{n}(i)\!\Rightarrow\! \stackrel{(n)}{y}_{\!\!j}\!\!(t)>0\mbox{ and} \stackrel{(\ell)}{y}_{\!\!j}\!\!(t)=0.
\end{equation}
\hfill$\small \blacksquare$
\end{theorem}
By Theorem~\ref{th:order2}, when island~$i$ is the only infected island in the network, infection at island~$j$ $n$-hops away from~$i$ is perturbed only through its $n$th-order derivative $\stackrel{(n)}{y}_{\!\!j}\!\!(t)$.
\begin{proof}
Again,
\begin{equation}
\stackrel{(n)}{y}_{\!\!j}\!\!(t)=\underbrace{\left(\gamma\sum_{k\sim j}\stackrel{(n-1)}{y}_{\!\!\!\!\!\!k}\!\!(t)\right)}_{\bf A}-\underbrace{\sum_{\ell=0}^{n-1}\left(\begin{array} {c} n-1\\\ell\end{array}\right)\stackrel{(\ell)}{y}_{\!\!j}\!\!(t)\left(\gamma\sum_{q\sim j}\stackrel{(n-1-\ell)}{y}_{\!\!\!\!\!\!\!\!q}\!\!(t)\right)}_{\bf B}-\underbrace{\stackrel{(n-1)}{y}_{\!\!\!\!\!\!j}\!\!(t)}_{\bf C}.\nonumber
\end{equation}
Now, we apply induction on the number of hops $n$.

\textit{Step~1}: For $n=1$, we have that $j\in\mathcal{N}(i)$ and
\begin{eqnarray}
\dot{y}_j(t) & = & \left(\gamma\sum_{k\sim j} y_k(t)\right)\left(1-y_j(t)\right)-y_j(t)=\gamma\sum_{k\sim j}y_k(t)\nonumber\\
 & = &\gamma y_i(t)+\gamma\underbrace{\sum_{k\sim j,k\neq i} y_k(t)}_{=0}= \gamma y_i(t)>0.\nonumber
\end{eqnarray}
That is, $y_j(t)=0$ and $\dot{y}_j(t)>0$.

\textit{Step~2}: \textbf{Induction step}. Assume assertion~\eqref{eq:order22} holds for $n-1$. We consider successively the terms~${\bf A}$, ${\bf B}$, and ${\bf C}$.

${\bf A}$: By definition, $j\in \mathcal{N}^{n}(i)\Rightarrow \exists\,\,{k\in \mathcal{N}^{n-1}\left(i\right)}:\,j\sim k$. From the induction hypothesis, $\stackrel{(n-1)}{y}_{\!\!\!\!\!\!k}\!\!(t)>0$. Therefore, $\gamma\sum_{k\sim j}\stackrel{(n-1)}{y}_{\!\!\!\!\!\!k}\!\!(t)>0$ and term ${\bf A}$ is strictly positive.

${\bf B}$: From Theorem~\ref{th:ordem}, for $j\in\mathcal{N}^{n}(i)$, $\stackrel{(\ell)}{y}_{\!\!j}\!\!(t)=0$, $\forall\,\,\,{\ell=1,\ldots,n-1}$, and, thus, term ${\bf B}$ is zero.

${\bf C}$: From Theorem~\ref{th:ordem}, for $j\in\mathcal{N}^{n}(i)$, $\stackrel{(n-1)}{y}_{\!\!\!\!\!\!j}\!\!(t)=0$, and term~${\bf C}$ is zero.

Therefore, $\stackrel{(n)}{y}_{\!\!j}\!\!(t)>0$ with $\stackrel{(\ell)}{y}_{\!\!j}\!\!(t)=0$, $\,\forall\,\,\,{\ell< n}$, and the Theorem is proved.
%\hfill$\small \blacksquare$
\end{proof}

 Theorems~\ref{th:ordem} and~\ref{th:order2} reveal the impact of the super-topology on the inter-dependence among the geometric aspects (e.g., derivative, curvature) of the infected populations across the islands, namely, a perturbation on the infected population of an island $i$ will perturb its immediate neighbors by perturbing their first derivatives. In general, for an $n$-hop geodesic connecting $i$ and $j$, we have that perturbations on $y_i(t)$ only affect the $n$th order curvature $\stackrel{(n)}{y}_{\!\!j}\!\!(t)$ in~$j$.

Next, we extend Theorem~\ref{th:invariant1} to regular multipartite networks, confirming that the state of infection of a regular multipartite network with a dominant initial degree of infection dominates the state of infection of other equivalent regular networks across the whole time $t$, $t\geq 0$.
\begin{theorem}
\label{th:invariant}
Let $\left(\mathbf{y}(t)\right)$ be the limiting macrostate in a regular multipartite network. Then,
\begin{equation}
\mathbf{y}(0)\leq \mathbf{y}_0 \Rightarrow \mathbf{y}\left(t,\mathbf{y}(0)\right)\leq \mathbf{y}\left(t,\mathbf{y}_0\right),\,\,\,\forall\,\,\,{t\geq 0}.\nonumber
\blacksquare
\end{equation}
%\hfill$\small \blacksquare$
\end{theorem}
\begin{proof}
%Similar to the proof of Theorem \ref{th:invariant1}
%\end{proof}
We show the invariance with respect to the dynamics~(\ref{eq:multipartite}) of the set
\begin{equation}
B\!=\!\left\{\!\left(\mathbf{y},\widetilde{\mathbf{y}}\right)\! \in\!\mathbb{R}^{2M}\!:y_i\!\geq \!\widetilde{y}_i,i=1,\ldots,M\!\right\}\cap \left[0,1\right]^{2M}.\nonumber
\end{equation}
Let $\left(\mathbf{y}(t,\mathbf{y_0}),\mathbf{y}(t,\mathbf{y}(0))\right)$ be the solution of
\begin{equation}
\frac{d}{dt}\left(\mathbf{y}(t),\mathbf{z}(t)\right)= \left(\mathbf{F}(\mathbf{y}(t)),\mathbf{F}(\mathbf{z}(t))\right).
\end{equation}
Then, it is enough to investigate the decoupled augmented vector field $\overline{\mathbf{F}}\left(\mathbf{y},\widetilde{\mathbf{y}}\right)=\left(\mathbf{F}(\mathbf{y}),\mathbf{F}(\widetilde{\mathbf{y}})\right)$ over the boundary of $B$ to establish that once started there, the solution $\left(\mathbf{y}(t,y_0),\mathbf{y}(t,\mathbf{y}(0))\right)$ never escapes the set $B$, i.e., $\mathbf{y}(t,\mathbf{y}_0)\geq \mathbf{y}(t,\mathbf{y}(0))$ for all $t$, $t\geq 0$, if $\mathbf{y}_0\geq \mathbf{y}(0)$. The set $B$ is depicted in Figure~\ref{fig:invariante2} as the Cartesian product of triangles and one has to assure that no solution components can leave the triangular regions. Let $t>0$ be such that:

\textit{Case~1}: $y_i(t)=1\,, 0< \widetilde{y}_i(t)< 1$:
\begin{equation}
F_i(\mathbf{y}(t))=\frac{d}{dt}y_i(t)=-y_i(t)=-1<0\nonumber
\end{equation}

\textit{Case~2}: $\widetilde{y}_i(t)=0\,, 0< y_i(t)< 1$:
\begin{equation}
F_i(\widetilde{\mathbf{y}}(t))=\frac{d}{dt}\widetilde{y}_i(t)=\gamma\sum_{j\sim i}y_j(t)>0\nonumber
\end{equation}

\textit{Case~3}: $0 \leq y_i(t)=\widetilde{y}_i(t)\leq 1$:
\begin{eqnarray}
F_i(\mathbf{y}(t))&=&\frac{d}{dt}y_i(t)=\left(\gamma\sum_{j\sim i}y_j(t)\right)\left(1-y_i(t)\right)-y_i(t)\nonumber\\
F_i(\widetilde{\mathbf{y}}(t))&=&\frac{d}{dt}\widetilde{y}_i(t)=\left(\gamma\sum_{j\sim i}\widetilde{y}_j(t)\right)\left(1-\widetilde{y}_i(t)\right)-\widetilde{y}_i(t)\nonumber\\
&=&\left(\gamma \sum_{j\sim i}\widetilde{y}_j(t)\right)\left(1-y_i(t)\right)-y_i(t).\label{eq:auxiliar}
\end{eqnarray}
If $y_j(0)\geq \widetilde{y}_j(0)$, $\forall\,\,\,{j\in\mathcal{N}(i)}$, with strict inequality for at least some $j\in\mathcal{N}(i)$, then $F_i(\mathbf{y}(0))=\dot{y}_i(0)>\dot{\widetilde{y}}_i(0)=F_i(\widetilde{\mathbf{y}}(0))$ and, therefore, from Theorem~\ref{th:analytic} and the analyticity of the vector field $\mathbf{F}$ (thus, the analyticity of the solutions), we have $y_j(t)> \widetilde{y}_j(t)$ for all $t\in\left(0,\epsilon\right)$ for some $\epsilon>0$ small enough. More generally, if $y_j(0)=\widetilde{y}_j(0)$, $\forall\,\,j\in\cup_{l=0}^{n-1}\mathcal{N}^l(i)$ for some $n\geq 2$ with $y_j(0)>\widetilde{y}_j(0)$ for some $j\in\mathcal{N}^l(i)$ with $l\geq n$ then, from Theorem~\ref{th:order2} it follows that  $\stackrel{(n)}{y}_{\!\!j}\!\!(0)>\stackrel{(n)}{\widetilde{y}}_{\!\!j}\!\!(0)$ and, thus, Theorem~\ref{th:analytic} yields $y_j(t)> \widetilde{y}_j(t)$ for all $t\in\left(0,\epsilon\right)$ for some $\epsilon>0$ small enough. Otherwise, if $y_j(0)=\widetilde{y}_j(0)$, $\forall j$, then, both $\left(\mathbf{y}(t)\right)$ and $\left(\widetilde{\mathbf{y}}(t)\right)$ obey the same differential equation with a Lipschitz continuous vector field over the compact domain $B$. The solution is thus unique and $y_j(t)=\widetilde{y}_j(t),\,\,\,\forall\,\,{t\geq 0}$. Figure~\ref{fig:invariante2} depicts the main idea of the proof.
\begin{figure}[hbt]
\begin{center}
\includegraphics[scale= 0.3]{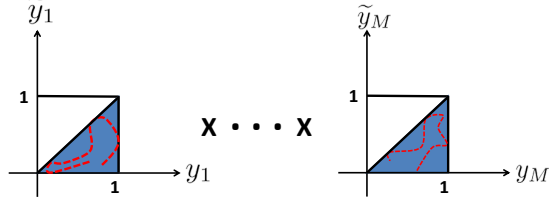}
\caption{Illustration of an orbit of the augmented system. The set $B$ is invariant, which implies that if $\mathbf{y}(0)\leq \widetilde{\mathbf{y}}(0)$ then $\mathbf{y}\left(t,\mathbf{y}(0)\right)\leq \widetilde{\mathbf{y}}\left(t,\widetilde{\mathbf{y}}(0)\right)$, $\forall\,\,{t\geq 0}$.}\label{fig:invariante2}
\end{center}
\end{figure}
%\hfill$\small \blacksquare$
\end{proof}

We state the main Theorem of the subsection on the ultimate condition on the microscopic parameter $\gamma$ that leads to the persistence of the virus in a $d$-regular multipartite network.
\begin{theorem}
\label{th:convergence}
Let the multipartite network be $d$-regular, i.e., each island is connected with $d$ other islands and let $\gamma$ be the inter-island transmission rate of the virus. If $\gamma>1$ and $\mathbf{y}(0)\neq 0$ then,
\begin{equation}
\left(y_i(t)\right)\longrightarrow \left(1-\frac{1}{d\gamma}\right),\,\,\,\forall\,\,\,{i=1,\ldots,M},\nonumber
\end{equation}
otherwise,
%\begin{equation}
$\left(\mathbf{y}(t)\right)\longrightarrow \mathbf{0}$.%\nonumber
%\end{equation}
\hfill$\small \blacksquare$
\end{theorem}
%\hfill$\small \blacksquare$
\begin{proof}
Let $y_i(0)>0$ and $\widetilde{y}_j(0):=\min\left\{y_i(0)\,:\,i=1,\ldots,M\right\}$, $\forall\,{j=1,\ldots,M}$. From Theorem~\ref{th:invariant}, $\widetilde{\mathbf{y}}(t):=\mathbf{y} \left(t,\widetilde{\mathbf{y}}(0)\right)\leq \mathbf{y}(t,\mathbf{y}(0))=:\mathbf{y}(t)$, $\forall\,{t\geq 0}$. Moreover, the solution of
\begin{equation}
\label{eq:complete}
\dot{\widetilde{y}}(t){\bf 1}_{M}=\left[d\gamma\widetilde{y}(t) \left(1-\widetilde{y}(t)\right)-\widetilde{y}(t)\right]{\bf 1}_{M}
\end{equation}
with
%\begin{equation}
$\widetilde{y}(0)=\min_{i=1,\ldots,M}\left\{\widetilde{y}_i(0)\right\}$, %\nonumber
%\end{equation}
is a solution to (\ref{eq:multipartite}), where ${\bf 1}_M\in\mathbb{R}^{M}$ is the vector of ones and $d$ is the degree of the regular super-network. Now, note that equation~(\ref{eq:complete}) captures the dynamics of the complete network. Thus,
\begin{equation}
\widetilde{y}_i(t)\longrightarrow \left(1-\frac{1}{d\gamma}\right),\,\,\,\forall\,\,\,{i=1,\ldots,M}.\nonumber
\end{equation}
Therefore,
\begin{equation}
\lim_{t\rightarrow+\infty} y_j(t)\geq \lim_{t\rightarrow+\infty} \widetilde{y}_j(t)=\left(1-\frac{1}{d\gamma}\right),\,\,\,\forall\,\,j.\nonumber
\end{equation}

It is left to prove the case where $y_j(0)=0$ for some $j$. We assume the worst scenario where $y_i(0)>0$ and $y_j(0)=0$, $\forall\,\,{j\neq i}$. Let $j\in\mathcal{N}^{k}(i)$. Then, from Theorem~\ref{th:order2}, $\stackrel{(k)}{y}_{\!\!j}\!\!(0)>0$, $\stackrel{(m)}{y}_{\!\!j}\!\!(0)=0$, $\forall\,\,{m< k}$.  Theorem~\ref{th:analytic} yields $y_j(t)> \widetilde{y}_j(t)$ for all $t\in\left(0,\epsilon\right)$ for some $\epsilon>0$ small enough. Now, $y_j(T^{\star})>0$, $\forall\,\,j$ for some $T^{\star}\in\left(0,\epsilon\right)$. Then, from the previous case (where we assumed $y_i(0)>0$, $\forall\,\,i$) we obtain that $y_j(t)\longrightarrow\left(1-\frac{1}{\gamma d}\right),\,\,\,\forall\,\,j$.
%\hfill$\small \blacksquare$
\end{proof}

In this Subsection, we provided a qualitative analysis of the mean field dynamics of the vector process $\left(\mathbf{y}(t)\right)$ over a regular multipartite network with equal sized islands. We proved in Theorems~\ref{th:ordem} and~\ref{th:order2} that the population in island $j\in\mathcal{N}^{n}(i)$ affects the dynamics of $y_i$ via its $n$th derivative, which connects the geometry of solutions with the underlying super-topology of the network. Then, we proved that lower/upperbounds on the initial conditions are preserved by the flow of our dynamics, i.e., $\mathbf{y}(0)\geq \mathbf{y}_0\Rightarrow \mathbf{y}\left(t,\mathbf{y}(0)\right)\geq \mathbf{y}\left(t,\mathbf{y}_0\right)$ for all $t\geq 0$. Then, we can squeeze any solution by symmetric well-characterized solutions to conclude that the virus resilience equilibrium state is a global attractor if $\gamma>1$. Otherwise, if $\gamma\leq1$, then $0$ is a global attractor state. In the next Subsection, we extend the analysis for the bi-viral case.

%\begin{figure} [hbt]
%\begin{center}
%\includegraphics[scale= 2]{Gauss.pdf}
%\caption{Pauca sed matura.}\label{fig:gauss}
%\end{center}
%\end{figure}

\subsection{Regular Multipartite Network: Multi-virus}\label{subsec:multipartitenet-bi-virus}
In this Subsection, we study the limiting dynamics of the spread of multiple strains of virus in a regular multipartite network starting by a bi-virus epidemics. We assume the symmetric configuration in Definition~\ref{def:symmetricconfiguration} where all islands have the same size and, therefore, the inter-island infection rates for each virus~$x$ and~$y$ are the same across the network, i.e., $\gamma^x_{ij}\equiv \gamma^x$, $\gamma^y_{ij}\equiv \gamma^y$, $\forall i\sim j$. We are particularly interested in determining the conditions to obtain a survival of the fittest type of phenomenon. The mean field dynamics for a bi-viral epidemics in a symmetric regular multipartite network are obtained from~\eqref{eq:multi}-\eqref{eq:multi2} by specializing them to a symmetric regular supernetwork:
\begin{eqnarray}
\frac{d}{dt}y_i(t) \!\!\!& = & \!\!\!\left(\gamma^y\sum_{j\sim i}y_j(t)\right)\left(1-x_i(t)-y_i(t)\right)-y_i(t):= F^y_i\left(\mathbf{y}(t),\mathbf{x}(t)\right),\:\:i=1,\cdots,M \label{eq:dyn_sym1}\\
\frac{d}{dt}x_i(t) \!\!\!& = & \!\!\! \left(\gamma^x\sum_{j\sim i}x_j(t)\right)\left(1-x_i(t)-y_i(t)\right)-x_i(t):= F^x_i\left(\mathbf{y}(t),\mathbf{x}(t)\right), \!\:\:i=1,\cdots,M\label{eq:dyn_sym2}
\end{eqnarray}
where we defined the vector field $\mathbf{F}\,:\,\left[0,1\right]^{2M}\rightarrow \mathbb{R}^{2M}$ as $\mathbf{F}=\left(F^y_1,\ldots,F^y_M,F^x_1,\ldots,F^x_M\right)$. The next Theorem is in line with Theorem~\ref{th:ordem} for single-virus spread and states that if two isomorphic regular super-networks $B_1$ and $B_2$ are evenly infected in a $n$-neighborhood around a supernode $i$, then the derivatives of $y_i$ and $x_i$ for the network $B_1$ coincide with the corresponding derivatives of $\widetilde{y}_i$ and $\widetilde{x}_i$ for the network $B_2$ up to an order $n$.
\begin{theorem}
\label{th:equal}
Let $y_i(0)=\widetilde{y}_i(0)$ and $x_i(0)=\widetilde{x}_i(0)$. Let $\overline{\mathcal{N}}^{n}(i)=\bigcup_{\ell=1}^{n}\mathcal{N}^{\ell}(i)$. Then:
\begin{align*}
&\left\{\begin{array}{l} y_k(0)=\widetilde{y}_k(0) \\ x_k(0)=\widetilde{x}_k(0)\end{array}\right. \forall\,\,\,{k\in \overline{\mathcal{N}}^{n}(i)} \:\:\:\Longrightarrow\\
&  \stackrel{(\ell)}{y}_{\!\!i}\!\!(0)=\stackrel{(\ell)}{\widetilde{y}}_{\!\!i}\!\!(0)\,\,\,\mbox{ and }\stackrel{(\ell)}{x}_{\!\!i}\!\!(0)=\stackrel{(\ell)}{\widetilde{x}}_{\!\!i}\!\!(0),\,\,\,\forall\,\,\,{\ell\leq n}.
\hspace{1cm}
\blacksquare
\end{align*}

\begin{align*}
&\left\{\begin{array}{l} y_k(0)=\widetilde{y}_k(0) \\ x_k(0)=\widetilde{x}_k(0)\end{array}\right. \forall\,\,\,{k\in \overline{\mathcal{N}}^{n}(i)} \:\:\:\Longrightarrow\\
&  \stackrel{(\ell)}{y}_{\!\!i}\!\!(0)=\stackrel{(\ell)}{\widetilde{y}}_{\!\!i}\!\!(0)\,\,\,\mbox{ and }\stackrel{(\ell)}{x}_{\!\!i}\!\!(0)=\stackrel{(\ell)}{\widetilde{x}}_{\!\!i}\!\!(0),\,\,\,\forall\,\,\,{\ell\leq n}.
\hspace{2cm}\blacksquare
\end{align*}

%\hfill$\small \blacksquare$
\end{theorem}
\begin{proof}
We apply induction on $n$. For $n=1$,
\begin{align*}
\frac{d}{dt}y_i(t)&=\left(\gamma^y\sum_{j\sim i}y_j(t)\right)\left(1-x_i(t)-y_i(t)\right)-y_i(t)\\
\frac{d}{dt}\widetilde{y}_i(t)&=\left(\gamma^y\sum_{j\sim i} \widetilde{y}_j(t)\right)\left(1-\widetilde{x}_i(t)- \widetilde{y}_i(t)\right)-\widetilde{y}_i(t).
\end{align*}
Note that $\widetilde{y}_j(0)=y_j(0),\,\,\,\forall\,\,\,{j\in\mathcal{N}(i)}$ and $x_i(0)=\widetilde{x}_i(0)$. By inspection, $\dot{\widetilde{y}}_i(0)=\dot{y}_i(0)$ and (by assumption) $y_i(0)=\widetilde{y}_i(0)$. Also, $x_i(0)=\widetilde{x}_i(0)$.

Now, assume Theorem~\ref{th:equal} holds for~$n-1$. We establish that it holds for $n$. We have:
\begin{eqnarray}
\stackrel{(n)}{y}_{\!\!i}\!\!(0) & = & \underbrace{\left(\gamma^y\sum_{j\sim i} \stackrel{(n-1)}{y}_{\!\!\!\!\!\!j}\!\!(0)\right)\left(1-y_i(0)-x_i(0)\right)}_{\bf A}\label{eq:induc11}\\& &-\underbrace{\sum_{\ell=1}^{n-1}\left(\begin{array} {c} n-1\\\ell\end{array}\right)\stackrel{(\ell)}{y}_{\!\!i}\!\!(0) \left(\gamma^y\sum_{q\sim i}\stackrel{(n-1-\ell)}{y}_{\!\!\!\!\!\!\!\!q}\!\!(0)\right)}_{\bf B}
-\underbrace{\stackrel{(n-1)}{y}_{\!\!\!\!\!\!i}\!\!(0)}_{\bf C}\nonumber\\
& &-\underbrace{\sum_{\ell=1}^{n-1}\left(\begin{array} {c} n-1\\\ell\end{array}\right)\stackrel{(\ell)}{x}_{\!\!i}\!\!(0) \left(\gamma^y\sum_{j\sim i}\stackrel{(n-1-\ell)}{y}_{\!\!\!\!\!\!\!\!j}\!\!(0)\right)}_{\bf D},\nonumber\\
%\end{eqnarray}
%and
%\begin{eqnarray}
\stackrel{(n)}{\widetilde{y}}_{\!\!i}\!\!(0) & = & \underbrace{\left(\gamma^y\sum_{j\sim i}\stackrel{(n-1)}{\widetilde{y}}_{\!\!\!\!\!\!j}\!\!(0)\right)\left(1-\widetilde{y}_i(0)-\widetilde{x}_{i}(0)\right)}_{\bf A}\label{eq:induc21}\\& &-\underbrace{\sum_{\ell=1}^{n-1}\left(\begin{array} {c} n-1\\\ell\end{array}\right)\stackrel{(\ell)}{\widetilde{y}}_{\!\!i}\!\!(0) \left(\gamma^y\sum_{q\sim i}\stackrel{(n-1-\ell)}{\widetilde{y}}_{\!\!\!\!\!\!\!\!q}\!\!(0)\right)}_{\bf B}
-\underbrace{\stackrel{(n-1)}{\widetilde{y}}_{\!\!\!\!\!\!i}\!\!(0)}_{\bf C}\nonumber\\
& &-\underbrace{\sum_{\ell=1}^{n-1}\left(\begin{array} {c} n-1\\\ell\end{array}\right)\stackrel{(\ell)}{\widetilde{x}}_{\!\!i}\!\!(0) \left(\gamma^y\sum_{j\sim i} \stackrel{(n-1-\ell)}{\widetilde{y}}_{\!\!\!\!\!\!\!\!j}\!\!(0)\right)}_{\bf D}.\nonumber
\end{eqnarray}
Recall the assumption
\begin{eqnarray}
\left\{\begin{array}{l} y_k(0)=\widetilde{y}_k(0) \\ x_k(0)=\widetilde{x}_k(0)\end{array}\right. \forall\,\,\,{k\in \overline{\mathcal{N}}^{(n)}(i)}.
\end{eqnarray}
 By induction, $ \forall\,\,\,{j\in\mathcal{N}(i)}$, $\forall \ell=1,\ldots,n-1$,  ${\stackrel{(\ell)}{y}_{\!\!j}\!\!(0)}= \stackrel{(\ell)}{\widetilde{y}}_{\!\!j}\!\!(0)$, $\stackrel{(\ell)}{x}_{\!\!j}\!\!(0)= \stackrel{(\ell)}{\widetilde{x}}_{\!\!j}\!\!(0)$ and also $\stackrel{(\ell)}{y}_{\!\!i}\!\!(0)= {\stackrel{(\ell)}{\widetilde{y}}_{\!\!i}\!\!(0)},\,\,\,\forall\,\,\,{\ell\leq n-1}$. Therefore, by inspection, we conclude that the terms ${\bf A}$, ${\bf B}$, ${\bf C}$, and ${\bf D}$ for both equations~\eqref{eq:induc11} and~\eqref{eq:induc21} match together, and, thus, $\stackrel{(n)}{y}_{\!\!j}\!\!(0)= \stackrel{(n)}{\widetilde{y}}_{\!\!j}\!\!(0)$. By symmetry, we also have that $\stackrel{(n)}{x}_{\!\!j}\!\!(0)= \stackrel{(n)}{\widetilde{x}}_{\!\!j}\!\!(0)$, and we conclude the proof of the Theorem.
%\hfill$\small \blacksquare$
\end{proof}

The next Theorem states that, if two regular multipartite systems have the same degree of infection at each island, except at some island $j$, $n$-hops away from island $i$, then there will be a mismatch between the $n$th-order derivative of the fraction of infected nodes at island $i$, $y_i$ and $\widetilde{y}_i$, in the two networks.

\begin{theorem}
\label{thm:multiunbalance}
Let $y_i(0)=\widetilde{y}_i(0)$ and $x_i(0)=\widetilde{x}_i(0)$
\begin{eqnarray}
\label{eq:assums1}
\left\{\begin{array}{l} y_k(0)=\widetilde{y}_k(0) \\ x_k(0)=\widetilde{x}_k(0)\end{array}\right. \forall\,\,\,{k\in \overline{\mathcal{N}}^{n-1}(i)}.
\end{eqnarray}
Also, let
\begin{eqnarray}
\label{eq:assums2}
\left\{\begin{array}{l} y_k(0)\geq \widetilde{y}_k(0) \\ x_k(0)\leq\widetilde{x}_k(0)\end{array}\right. \forall\,\,\,{k\in \mathcal{N}}^{n}(i)\setminus \left\{m\right\}
\end{eqnarray}
with strict inequality $y_{m}(0)>\widetilde{y}_{m}(0)$ for some $m\in\mathcal{N}^{n}(i)$. Then, $\stackrel{(n)}{y}_{\!\!i}\!\!(0)>\stackrel{(n)}{\widetilde{y}}_{\!\!i}\!\!(0)$.
\hfill$\small \blacksquare$
\end{theorem}

\begin{proof}

We apply induction on the number of hops $n$.

\textit{Case~1}: For $n=1$, from the assumptions of the Theorem, we conclude $\dot{y}_i(0)>\dot{\widetilde{y}}_i(0)$ since
\begin{align*}
\dot{y}_i(0){}&=\left(\gamma^y\sum_{j\sim i} y_j(0)\right)\left(1-x_i(0)-y_i(0)\right)-y_i(0)
\\
{}&>\left(\gamma^y\sum_{j\sim i}\widetilde{y}_j(0)\right)\left(1-x_i(0)-y_i(0)\right)-y_i(0)
\\
{}&=\dot{\widetilde{y}}_i(0).
\end{align*}

\textit{Case~2}: \textbf{Induction step}. Assume that Theorem~\ref{thm:multiunbalance} holds for $n-1$ and let us prove that it holds for $n$. We consider successively the terms~${\bf A}$, ${\bf B}$, ${\bf C}$, and ${\bf D}$ in equations~\eqref{eq:induc11} and~\eqref{eq:induc21}.

${\bf A}$: Note that for some $j\in\mathcal{N}(i)$ we have that $m\in\mathcal{N}^{(n-1)}(j)$ where $m$ is defined in the assumptions of the Theorem. Thus, by the induction hypothesis, we have $\stackrel{(n-1)}{y}_{\!\!\!\!\!\!j}\!\!(0)> {\stackrel{(n-1)}{\widetilde{y}}_{\!\!\!\!\!\!j}\!\!(0)}$, and, hence, the term ${\bf A}$ in equation~\eqref{eq:induc11} is greater than its counterpart in equation~\eqref{eq:induc21}.

${\bf B}$ and ${\bf C}$: It should be now clear that these terms match together between equations~\eqref{eq:induc11} and~\eqref{eq:induc21}.

${\bf D}$: From Theorem~\ref{th:equal}, it follows that $\stackrel{(\ell)}{x}_{\!\!j}\!\!(0)=\stackrel{(\ell)}{\widetilde{x}}_{\!\!j}\!\!(0)$ for all $\ell=1,\ldots,n-1$ and thus, term ${\bf D}$ is the same for both equations.

Therefore, $\stackrel{(n)}{y}_{\!\!j}\!\!(0)>\stackrel{(n)}{\widetilde{y}}_{\!\!j}\!\!(0)$ and the Theorem is proved.
%\hfill$\small \blacksquare$
\end{proof}

The next Theorem is an extension of the monotonous property for a single virus spread established in Theorem~\ref{th:invariant} to the bi-viral epidemics case: appropriate bounds on the initial conditions are preserved by the flow of the dynamical system~(\ref{eq:dyn_sym1})-(\ref{eq:dyn_sym2}).
\begin{theorem}\label{thm:invariant-bi-virus}
%If $\widetilde{y}_j(0)\leq y_j(0)\,,\widetilde{x}_j(0)\geq x_j(0),\,\,\,\forall\,\,\,{j\in I},\,\,\,\,$ then $\widetilde{y}_j(t)\leq
%y_j(t)\,,\widetilde{x}_j(t)\geq x_j(t),\,\,\,\forall\,\,\,{j\in I},\,\,\,\forall\,\,\,{t\leq0}$,
If $\mathbf{y}(0)\leq \mathbf{y}_0\,,\mathbf{x}(0)\geq \mathbf{x}_0$ then, $\mathbf{y}(t,\mathbf{z}(0))\leq \mathbf{y}(t,\mathbf{z}_0)\,,\mathbf{x}(t,\mathbf{z}(0))\geq \mathbf{x}(t,\mathbf{z}_0)$, where we define
%\begin{eqnarray}
%\mathbf{z}(t)&:=&\left(\mathbf{x}(t),\mathbf{y}(t)\right)\nonumber\\
%\mathbf{z}_0&:=&\left(\mathbf{x}_0,\mathbf{y}_0\right)\nonumber
%\end{eqnarray}
$\mathbf{z}(t):=\left(\mathbf{x}(t),\mathbf{y}(t)\right)$ and
$\mathbf{z}_0:=\left(\mathbf{x}_0,\mathbf{y}_0\right)$.
\hfill$\small \blacksquare$
\end{theorem}

\begin{proof}
Assume that $\mathbf{y}(0)\neq \mathbf{y}_0$ or $\mathbf{x}(0)\neq \mathbf{x}_0$, otherwise, from uniqueness, the solutions are equal. Define
\begin{equation}
T=\inf\left\{t\,:\,t\geq 0,\,  \mathbf{y}\left(t,\mathbf{z}(0)\right)\nleq \mathbf{y}\left(t,\mathbf{z}_0\right)\mbox{  or  } \mathbf{x}\left(t,\mathbf{z}(0)\right)\ngeq \mathbf{x}\left(t,\mathbf{z}_0\right)\right\}.\nonumber
\end{equation}
Assume that $T<\infty$. Then, for $i,j\in\left\{1,\ldots,M\right\}$ with $i\neq j$, we have one of the following:
\begin{align}
\label{eq:first-multipartite}
y_i\left(T,\mathbf{z}(0)\right){}&=y_i\left(T,\mathbf{z}_0\right)\mbox{  and  }y_j\left(T,\mathbf{z}(0)\right)<y_j\left(T,\mathbf{z}_0\right)
\\
%\end{align}
%or
%\begin{align}
\label{eq:firstb-multipartite}
x_i\left(T,\mathbf{z}(0)\right){}&=x_i\left(T,\mathbf{z}_0\right)\mbox{  and  }x_j\left(T,\mathbf{z}(0)\right)>x_j\left(T,\mathbf{z}_0\right).
\end{align}
\begin{figure} [hbt]
\begin{center}
\includegraphics[scale= 0.3]{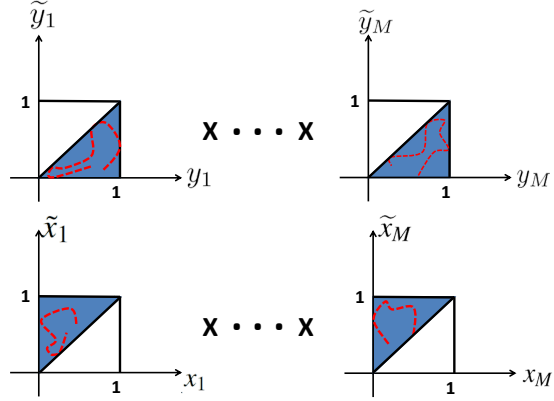}
\caption{Representation of the phase space of the state of the augmented dynamical system $\left(\!\widetilde{x}_1\!(t),x_1\!(t),\widetilde{y}_1\!(t),y_1\!(t),\!\ldots, \widetilde{x}_M\!(t),x_M\!(t),\widetilde{y}_M\!(t),y_M\!(t)\!\right)$, i.e., the system capturing the evolution of the two virus in isomorphic regular multipartite networks with perhaps different initial degrees of infection. The red (dashed) curve captures the idea that a solution cannot escape the blue (dark) region in finite time.}\label{fig:invariants}
\end{center}
\end{figure}
Wlog, choose configuration~\eqref{eq:first-multipartite} and assume $j\in\mathcal{N}^n(i)$ is the closest island to $i$ where we have strict inequality $y_j(T,\mathbf{z}(0))<y_j(T,\mathbf{z}_0)$.

\textit{Case~1}: If $x_1\left(T, \mathbf{z}(0)\right)+y_1\left(T,\mathbf{z}(0)\right)<1$, then, from equations~\eqref{eq:dyn_sym1} and~\eqref{eq:dyn_sym2}, and from Theorem~\ref{thm:multiunbalance} we have

\begin{equation}
\stackrel{(n)}{y}_{\!\!i}\!\!(T,\mathbf{z}(0))<\stackrel{(n)}{\widetilde{y}}_{\!\!i}\!\!(T,\mathbf{z}_0).\nonumber
\end{equation}

Therefore, from Theorem~\ref{th:analytic} in the Appendix, we have that
\begin{equation}
\exists\,\,\epsilon_1>0:\,y_1\left(t,\mathbf{z}(0)\right)<y_1\left(t,\mathbf{z}_0\right),\,\,\,\:\forall\,\,\, T<t<T+\epsilon_1.\nonumber
\end{equation}
Also,
\begin{equation}
y_j\left(T,\mathbf{z}(0)\right)<y_j\left(T,\mathbf{z}_0\right)\Rightarrow \exists\,\,\epsilon_2>0:\,y_j\left(t,\mathbf{z}(0)\right)<y_j\left(t,\mathbf{z}_0\right),\,\,\,\:\forall\,\,\, T<t<T+\epsilon_2.\nonumber
\end{equation}
Thus,
\begin{equation}
y\left(t,\mathbf{z}(0)\right)\leq y\left(t,\mathbf{z}_0\right),\,\,\,\forall\,\,\,\: T<t<T+\epsilon\nonumber
\end{equation}
with $\epsilon=\epsilon_1\wedge \epsilon_2$. Similarly, we have that
\begin{equation}
x\left(t,\mathbf{z}(0)\right)\geq x\left(t,\mathbf{z}_0\right),\,\,\,\:\forall\,\,\, T<t<T+\alpha\nonumber
\end{equation}
for some $\alpha>0$.

\textit{Case~2}: If $x_1\left(T,\mathbf{z}(0)\right)+y_1\left(T,\mathbf{z}(0)\right)=1$, then,
\begin{equation}
\dot{x}_1\left(T,\mathbf{z}(0)\right)+\dot{y}_1\left(T,\mathbf{z}(0)\right)=-\left(x_1(T)+y_1(T)\right)<0\Rightarrow\exists\,\, \epsilon>0:\,x_1\left(t,\mathbf{z}(0)\right)+y_1\left(t,\mathbf{z}(0)\right)<1,\nonumber
\end{equation}
for all $t\in\left(T,T+\epsilon\right)$. In any case, we reach a contradiction on the definition of $T$, and the Theorem is proved.
%\hfill$\small \blacksquare$
\end{proof}

Now, through similar arguments as in the previous Subsections, one can bound any solution by symmetrically initialized solutions, leading to the next Theorem.

\begin{theorem}\label{th:converg}
Let $\left(\mathbf{x}(t), \mathbf{y}(t)\right)$ be solution of the following bi-viral limiting dynamics over a regular multipartite network:
\begin{eqnarray}
\frac{d}{dt}y_i(t) & = & \left(\gamma^{y}\sum_{j\sim i}y_j(t)\right)\left(1-x_i(t)-y_i(t)\right)-y_i(t)\label{eq:dynam}\\
\frac{d}{dt}x_i(t) & = & \left(\gamma^{x}\sum_{j\sim i}x_j(t)\right)\left(1-x_i(t)-y_i(t)\right)-x_i(t).\label{eq:dynam2}
\end{eqnarray}
Let $\gamma^{x}>\gamma^{y}$. If $\gamma^{x}>\frac{1}{d}$ then,
\begin{eqnarray}
\mathbf{x}(t) & \longrightarrow & \left(1-\frac{1}{\gamma^x d}\right)\mathbf{1}_M\nonumber \\
\mathbf{y}(t) & \longrightarrow & \mathbf{0}\nonumber
\end{eqnarray}
otherwise
%\begin{eqnarray}
$\mathbf{x}(t) \longrightarrow \mathbf{0}$ %\nonumber \\
and
$\mathbf{y}(t) \longrightarrow \mathbf{0}$.%\nonumber
%\end{eqnarray}
\hfill$\small \blacksquare$
\end{theorem}

\begin{proof}
We first consider the solutions symmetrically initialized, $\mathbf{y}(0)=\alpha_1\mathbf{1}_M$ and $\mathbf{x}(0)=\alpha_2\mathbf{1}_M$, which turn out to be also solutions for the reduced system:
\begin{eqnarray}
\frac{d}{dt}y(t)\mathbf{1}_M & = & \left(\gamma^{y}d y(t)\left(1-x(t)-y(t)\right)-y(t)\right)\mathbf{1}_M\nonumber\\%\label{eq:miu}\\
\frac{d}{dt}x(t)\mathbf{1}_M & = & \left(\gamma^{x}d x(t)\left(1-x(t)-y(t)\right)-x(t)\right)\mathbf{1}_M.\nonumber%\label{eq:miu2}
\end{eqnarray}
The equations also describe the dynamics of diffusion of two virus in a complete network explored in Reference~\cite{paper:CDC}. Thus, if $\gamma^{x}>\gamma^{y}$ with $\gamma^{x}>\frac{1}{d}$, we have
\begin{eqnarray}
\mathbf{x}(t) & \longrightarrow & \left(1-\frac{1}{\gamma^x d}\right)\mathbf{1}_M\nonumber \\
\mathbf{y}(t) & \longrightarrow & \mathbf{0}\nonumber
\end{eqnarray}
otherwise, if $\gamma^{x}\leq\frac{1}{d}$,
%\begin{eqnarray}
$\mathbf{x}(t) \longrightarrow \mathbf{0}$ %\nonumber \\
and
$\mathbf{y}(t)  \longrightarrow \mathbf{0}$.%\nonumber
%\end{eqnarray}

For general solutions other than symmetrically initialized, a bound argument squeezes any solution by these simpler ones, resorting to Theorem~\ref{thm:invariant-bi-virus} similarly to as done for the bipartite network with two virus spread.
\end{proof}

The next Theorem finally states that among many distinct strains of virus in a symmetric regular multipartite network, only the strongest strain eventually survives and all the remaining weaker ones die out. The ODE~(\ref{eq:dynammultivi}) is the corresponding meanfield dynamics obtained from the peer-to-peer rules of infection in the limit of large networks,\cite{augusto_moura_emergent}. In what follows, we refer to $y_{ik}(t)$ as the fraction of $k$-infected nodes at island $i$ at time $t\geq 0$.

\begin{theorem}
Let $\left(\mathbf{y}(t)\right)$ be solution of the following multi-virus limiting dynamics over a symmetric $d$-regular multipartite network:
\begin{equation}
\frac{d}{dt}y_{ik}(t)  = \left(\gamma^{k}\sum_{j\sim i}y_{jk}(t)\right)\left(1-\sum_{\ell=1}^{K}y_{i\ell}(t)\right)-y_{ik}(t).\label{eq:dynammultivi}
\end{equation}
Let $k^{\star}$ be the most virulent strain, i.e., $\gamma^{k^{\star}}>\gamma^{k}$ for all $k\neq k^{\star}$. Define $\left(\mathbf{y}_k(t)\right)=\left(y_{1k}(t),\ldots,y_{Mk}(t)\right)$, as collecting the fraction of $k$-infected nodes across islands. If $\gamma^{k^{\star}}>\frac{1}{d}$ then, for all $k\neq k^{\star}$
\begin{eqnarray}
\mathbf{y}_{k^{\star}}(t) & \longrightarrow & \left(1-\frac{1}{\gamma^{k^{\star}} d}\right)\mathbf{1}_M\nonumber \\
\mathbf{y}_{k}(t) & \longrightarrow & \mathbf{0}\nonumber
\end{eqnarray}
otherwise, if $\gamma^{k^{\star}}\leq\frac{1}{d}$, then, 
%\begin{eqnarray}
$\mathbf{y}(t) \longrightarrow \mathbf{0}$ %\nonumber \\
%\end{eqnarray}
\hfill$\small \blacksquare$
\end{theorem}

\begin{proof}
First, it is easy to check that if $\left(\mathbf{y}(t)\right)$ is solution of the ODE~(\ref{eq:dynammultivi}) and if $\mathbf{y}_k(0)=\mathbf{0}$ for some $k\in \left\{1,\ldots,K\right\}$ then, $\mathbf{y}_k(t)=\mathbf{0}$ for all time $t\geq 0$. In words, if a virus strain is not present in the network at time $\overline{t}\geq 0$ then, it will remain extinct for all future times $t\geq \overline{t}$. Now, let
\begin{eqnarray}
\left\{\begin{array}{lll} \mathbf{y}_{k^{\star}}(0) & \geq & \mathbf{\widetilde{y}}_{k^{\star}}(0)\\ \sum_{k\neq k^{\star}}\mathbf{y}_{k}(0) & \leq & \mathbf{\widetilde{y}}_{i\hat{k}}(0)\end{array}\right..\label{eq:multiupp}
\end{eqnarray}
The inequalities above are preserved by the dynamics
\begin{eqnarray}
\left\{\begin{array}{lll} \mathbf{y}_{k^{\star}}(t) & \geq &  \mathbf{\widetilde{y}}_{k^{\star}}(t) \\ \sum_{k\neq k^{\star}} \mathbf{y}_{k}(t) & \leq & \mathbf{\widetilde{y}}_{\hat{k}}(t)\end{array}\right.,\label{eq:multiupp2}
\end{eqnarray}
for all $t\geq 0$, where $\left(\mathbf{y}(t)\right)$ and $\left(\mathbf{\widetilde{y}}(t)\right)$ are solutions of~(\ref{eq:dynammultivi}) with initial conditions $\mathbf{y}(0)$ and $\mathbf{\widetilde{y}}(0)$ obeying inequalities~(\ref{eq:multiupp}). We can establish this fact through similar invariance type of arguments as, for instance, in the proof of Theorem~\ref{thm:invariant-bi-virus}: let $T$ be the hitting time to invalidate any of the inequalities in equation~(\ref{eq:multiupp2}), assume that $T<\infty$ and reach a contradiction (we do not repeat the steps here). Let $\hat{k}$ be the second strongest strain, i.e., $\gamma^{k}<\gamma^{\hat{k}}<\gamma^{k^{\star}}$ for all $k\neq \hat{k}$ and $k\neq k^{\star}$. For any initial condition $\mathbf{y}(0)=\mathbf{y}_0\in \left[0,1\right]^{M\times K}$, we can choose $\mathbf{\widetilde{y}}(0)\in \left[0,1\right]^{M\times K}$, with $\mathbf{\widetilde{y}}_{k}(0)\neq0$, if $k=k^{\star}$ or $k=\hat{k}$ and, $\mathbf{\widetilde{y}}_{k}(0)=0$ otherwise, so that $\mathbf{y}(0)$ and $\mathbf{\widetilde{y}}(0)$ obey inequalities~(\ref{eq:multiupp}). In this case,
\begin{eqnarray}
\left\{\begin{array}{lll} \mathbf{y}_{k^{\star}}(t) & \geq & \mathbf{\widetilde{y}}_{k^{\star}}(t)\rightarrow \left(1-\frac{1}{\gamma^{k^{\star}}d}\right)\mathbf{1}\\ \sum_{k\neq k^{\star}}\mathbf{y}_{k}(t) & \leq & \mathbf{\widetilde{y}}_{\hat{k}}(t)\rightarrow \mathbf{0}\end{array}\right.
\end{eqnarray}
from Theorem~\ref{th:converg} and since $\left(\mathbf{\widetilde{y}}_{k^{\star}}(t),\mathbf{\widetilde{y}}_{\hat{k}}(t)\right)$ is solution of~(\ref{eq:dynam})-(\ref{eq:dynam2}), that is, $\left(\mathbf{\widetilde{y}}_{k^{\star}}(t),\mathbf{\widetilde{y}}_{\hat{k}}(t)\right)$ refers to the evolution of two strains $k^{\star}$ and $\hat{k}$ and from Theorem~\ref{th:converg} the strongest $k^{\star}$ may survive and the weaker one $\hat{k}$ dies out.

Similarly, by considering the weakest strain $w$, i.e., $\gamma^{w}<\gamma^{k}$ for all $k\in \left\{1,\ldots,K\right\}$, the inequalities~(\ref{eq:multiupp})-(\ref{eq:multiupp2}) can be reverted,
\begin{eqnarray}
\left\{\begin{array}{lll} \mathbf{y}_{k^{\star}}(0) & \leq & \mathbf{\hat{y}}_{k^{\star}}(0)\\ \sum_{k\neq k^{\star}}\mathbf{y}_{k}(0) & \geq & \mathbf{\hat{y}}_{w}(0)\end{array}\right. & \Rightarrow &
\left\{\begin{array}{lll} \mathbf{y}_{k^{\star}}(t) & \leq &  \mathbf{\hat{y}}_{k^{\star}}(t) \rightarrow \left(1-\frac{1}{\gamma^{k^{\star}}d}\right)\mathbf{1}\\ \sum_{k\neq k^{\star}} \mathbf{y}_{k}(t) & \geq & \mathbf{\hat{y}}_{w}(t)\rightarrow \mathbf{0}\end{array}\right.. \nonumber
\end{eqnarray}
To sum up, for any initial condition $\mathbf{y}(0)$, we can choose $\mathbf{\widetilde{y}}(0),\mathbf{\hat{y}}(0)\in\left[0,1\right]^{M\times K}$ so that 
\begin{eqnarray}
\mathbf{\widetilde{y}}_{k^{\star}}(t)  \leq  \mathbf{y}_{k^{\star}}(t)  \leq  \mathbf{\hat{y}}_{k^{\star}}(t)\nonumber\\ 
\mathbf{\hat{y}}_{w}(t)  \leq  \sum_{k\neq k^{\star}} \mathbf{y}_{k}(t)\leq \mathbf{\widetilde{y}}_{\hat{k}}(t)\nonumber
\end{eqnarray}
for all $t\geq 0$ and thus,
\begin{eqnarray}
\mathbf{y}_{k^{\star}}(t) & \rightarrow & \left(1-\frac{1}{\gamma^{k^{\star}}d}\right)\mathbf{1}\nonumber\\
\sum_{k\neq k^{\star}} \mathbf{y}_{k}(t) & \rightarrow & \mathbf{0}.\nonumber
\end{eqnarray}
It is easy to check that the set $\left[0,1\right]^{M\times K}$ is invariant under the dynamics~(\ref{eq:dynammultivi}). The Theorem is now proved.
\end{proof}

%%%%%%%%%%%%%%%%%%%%%%%%%%%%%%%%%%%%%%%%%%%%%%%%%%%%%%%%%%%%%%%%%%%%%%%%%%%%%%%%%%%%%%%%%%%%%%%%%%%%%%%%%%%%%%%%%%%%%%%%%%%%%%%%%%%%%%%%%%%%%%%%%%%%%%%%%%%%%%%%%
\section{Concluding Remarks}\label{sec:asymmetry}
\hspace{0.43cm}
There are three issues in determining the macroscopic behavior in stochastic networks:
\begin{inparaenum}[1)]
\item finding a Markovian macrostate, i.e., low dimensional functionals of the microstate $\mathbf{X}^N(t)$ that are Markov;
\item deriving the equations for the dynamics of the macrostate in the limit of large networks--the mean field dynamics of the macrostate; and
\item studying the qualitative dynamics of the mean field. The first and second items are dealt with in~\cite{augusto_moura_emergent}; the third is our concern here.
  \end{inparaenum}

We analyzed the limiting (in the number of nodes) dynamics of a virus spreading in a regular multipartite network. Our method to derive the qualitative analysis of such coupled nonlinear dynamical system is not Lyapunov theory nor numerical simulations based. Instead, we explored a monotonous structure of the system, upper/lower bounding by simpler solutions any solution of the mean field equations. Our main conclusions for symmetric generic regular multipartite networks are:
\begin{enumerate}
\item \textit{Virus Resilience}: If $\gamma >\frac{1}{d}$, the virus persists in the network; otherwise, it dies out.
\item \textit{Natural Selection--Survival of the Fittest}: Only one strain (the most virulent one) survives, the remaining weaker ones die out; if $\gamma^{k^{\star}} >\gamma^k$ for all $k\neq k^{\star}$ with $\gamma^{k^{\star}}>\frac{1}{d}$, then virus~$k^{\star}$ persists in the network and all the remaining strains die out.
\end{enumerate}
 For general multipartite networks, the break of symmetry may defy natural selection; this is bing pursued in our current research.

\appendix
\section{Appendix}

\begin{theorem}
\label{th:analytic}
Let $f\,:\,\left(0,+\infty\right)\rightarrow \mathbb{R}$ be an analytic function. If for some $T\in\mathbb{R}$ we have $\stackrel{(k)}{f}{\!\!\!\!}(T)>0$, $\stackrel{(m)}{f}{\!\!\!\!}(T)=0$, $\forall\,\,{m=0,1,\ldots,k-1}$ and $k\geq 1$ then, there exists $\epsilon>0$ such that $f(t)>0$ for all $t\in\left(T,T+\epsilon\right)$.
\end{theorem}

\begin{proof}
Without loss of generality, assume $T=0$. Since $f\in C^{\omega}(\mathbb{R})$ then,
\begin{equation}
f(t) = f(0)+\dot{f}(0)t+\ddot{f}(0)t^2+\ldots+\stackrel{(k)}{f}{\!\!}(0)t^k+r(t) =  f(0)+\left( \stackrel{(k)}{f}{\!\!}(0)+\frac{r(t)}{t^k}\right)t^k,\nonumber
\end{equation}
with $\frac{|r(t)|}{t^k}\rightarrow 0$ as $t\rightarrow 0$. Choose $\delta$ such that $\frac{|r(t)|}{t^k}< \frac{\stackrel{(k)}{f}{\!}(0)}{2}$, $\forall\,\,\,{t\in\left(0,\delta\right)}$. Then,
\begin{equation}
\stackrel{(k)}{f}{\!\!}(0)+\frac{r(t)}{t^k}>0,\,\,\, \forall\,\,\,{t\in\left(0,\delta\right)}.\nonumber
\end{equation}
Then,
\begin{equation}
f(t)=f(0)+\left(\stackrel{(k)}{f}{\!\!}(0)+\frac{r(t)}{t^k}\right)t^k>0,\,\,\, \forall\,\,\,{t\in\left(0,\delta\right)}.\nonumber
\end{equation}
\end{proof}

%\hfill$\small \blacksquare$

%\begin{thebibliography}%{9}

\small
\bibliographystyle{IEEEtran}
\bibliography{IEEEabrv,biblio}

\begin{thebibliography}{10}
\providecommand{\url}[1]{#1}
\csname url@samestyle\endcsname
\providecommand{\newblock}{\relax}
\providecommand{\bibinfo}[2]{#2}
\providecommand{\BIBentrySTDinterwordspacing}{\spaceskip=0pt\relax}
\providecommand{\BIBentryALTinterwordstretchfactor}{4}
\providecommand{\BIBentryALTinterwordspacing}{\spaceskip=\fontdimen2\font plus
\BIBentryALTinterwordstretchfactor\fontdimen3\font minus
  \fontdimen4\font\relax}
\providecommand{\BIBforeignlanguage}[2]{{%
\expandafter\ifx\csname l@#1\endcsname\relax
\typeout{** WARNING: IEEEtran.bst: No hyphenation pattern has been}%
\typeout{** loaded for the language `#1'. Using the pattern for}%
\typeout{** the default language instead.}%
\else
\language=\csname l@#1\endcsname
\fi
#2}}
\providecommand{\BIBdecl}{\relax}
\BIBdecl

\bibitem{augusto_moura_emergent}
A.~Santos, J.~M.~F. Moura, and J.~M.~F. Xavier, ``Emergent behavior in
  multipartite large networks: Multi-virus epidemics,'' 2013, submitted.
  http://arxiv.org/abs/1306.6198.

\bibitem{Daley}
D.~J. Daley and J.~Gani, \emph{Epidemic Modelling: An Introduction}.\hskip 1em
  plus 0.5em minus 0.4em\relax Cambridge, UK: Cambridge University Press, 2001.

\bibitem{Antunes}
N.~Antunes, C.~Fricker, P.~Robert, and D.~Tibbi, ``Analisys of loss networks
  with routing,'' \emph{The Annals of Applied Probability}, vol.~16, no.~4, pp.
  2007--2026, 2006.

\bibitem{Antunes2}
------, ``Stochastic networks with multiple stable points,'' \emph{The Annals
  of Applied Probability}, vol.~36, no.~1, pp. 255--278, 2008.

\bibitem{Mieghem}
P.~Van~Mieghem, J.~Omic, and R.~Kooij, ``Virus spread in networks,''
  \emph{Networking, IEEE/ACM Transactions on}, vol.~17, no.~1, pp. 1 --14, feb.
  2009.

\bibitem{paper:CDC}
A.~Santos and J.~M.~F. Moura, ``Emergent behavior in large scale networks,'' in
  \emph{2011 50th IEEE Conference on Decision and Control and European Control
  Conference (CDC-ECC)}, December 2011, pp. 4485 --4490.

\bibitem{Pastor-Satorras-Vespignani-2001}
\BIBentryALTinterwordspacing
R.~Pastor-Satorras and A.~Vespignani, ``Epidemic spreading in scale-free
  networks,'' \emph{Phys. Rev. Lett.}, vol.~86, pp. 3200--3203, Apr 2001.
  [Online]. Available:
  \url{http://link.aps.org/doi/10.1103/PhysRevLett.86.3200}
\BIBentrySTDinterwordspacing

\bibitem{particle}
P.~Donnellya and D.~Welsha, ``Finite particle systems and infection models,''
  \emph{Mathematical Proceedings of the Cambridge Philosophical Society},
  vol.~94, no.~1, pp. 167--182, July 1983.

\bibitem{Jackson}
M.~O. Jackson, \emph{Social and Economic Networks}.\hskip 1em plus 0.5em minus
  0.4em\relax Princeton University Press, 2008.

\bibitem{Paper:CDC-2012}
A.~Santos and J.~M.~F. Moura, ``Diffusion and topology: Large densely connected
  bipartite networks,'' in \emph{2012 51th IEEE Conference on Decision and
  Control and European Control Conference (CDC-ECC)}, December 2012, pp. 4485
  --4490.

\end{thebibliography}
%

%\end{thebibliography}

\end{document}